\documentclass[a4paper]{article}

\usepackage[top=2.5cm, bottom=2.5cm, left=2.5cm , right=2.5cm]{geometry}
\usepackage{hyperref}
\hypersetup{colorlinks,
	citecolor=red, 
	filecolor=black,
	linkcolor=blue,
	urlcolor=black}
\usepackage{lmodern}
\usepackage[english]{babel}
\usepackage[utf8x]{inputenc}
\usepackage{color, colortbl}
\usepackage{graphicx}
\usepackage{amssymb, amsmath, amsfonts, amsxtra, amsthm}
\usepackage{subfigure}
\usepackage{mathrsfs}
\usepackage{mathtools}

\usepackage{tikz}
\usetikzlibrary{arrows, calc, shapes, positioning, arrows.meta}
\usetikzlibrary{decorations.markings, decorations.pathmorphing, decorations.pathreplacing}
\usetikzlibrary{patterns, shapes.geometric}

\newtheorem{theo}{Theorem}

\newtheorem{lemm}{Lemma}
\newtheorem{prop}{Proposition}

\newtheorem{rema}{Remark}

\makeatletter
\renewcommand\section{\@startsection{section}{1}{\z@}%
	{-3.5ex \@plus -3ex \@minus -.2ex}%
	{2.3ex \@plus .2ex}%
	{\bfseries\centering}}% 
\makeatother

\makeatletter
\renewcommand\subsection{\@startsection{subsection}{1}{\z@}%
	{-3.5ex \@plus -3ex \@minus -.2ex}%
	{2.3ex \@plus .2ex}%
	{\flushleft\textit}}% 
\makeatother

\newcommand{\abs}[1]{\left\vert#1\right\vert}

\newcommand{\R}{\mathbb{R}}

\newcommand{\diag}[1]{\text{diag}\left\{#1\right\}}

\newcounter{constant}

\begin{document}

\title{Influence of the topology on the dynamics of a complex network of HIV/AIDS epidemic models}

\author{Guillaume Cantin$^1$\\
	\texttt{guillaumecantin@mail.com}
	\and
	Cristiana J. Silva$^2$\\
	\texttt{cjoaosilva@ua.pt}
}

\date{$1$ - Laboratoire de Math\'ematiques Appliqu\'ees du Havre, Normandie Universit\'e,\\
	FR CNRS 3335, ISCN, 76600 Le Havre, France\\
	$2$ - Center for Research and Development in Mathematics and Applications (CIDMA),\\
	Department of Mathematics, University of Aveiro, 3810-193 Aveiro, Portugal}

\begin{abstract}
In this paper, we propose an original complex network model for an epidemic problem in an heterogeneous geographical area.
The complex network is constructed by coupling nonidentical instances of a HIV/AIDS epidemiological model for which a disease-free equilibrium and an endemic equilibrium can coexist.
After proving the existence of a positively invariant region for the solutions of the complex network problem, we investigate the effect of the coupling on the dynamics of the network, and establish the existence of a unique disease-free equilibrium for the whole network, which is globally asymptotically stable.
We prove the existence of an optimal topology that minimizes the level of infected individuals, and apply the theoretical results to the case of the Cape Verde archipelago.
\end{abstract}

\textbf{keywords:} Complex network, epidemiological model, basic reproduction number, graph topology, HIV/AIDS, Cape Verde
\newline

\maketitle

\section{Introduction}

In the last decades, mathematical models have played a very important role on the study of 
the analysis of the spread of infectious diseases.
One of those infectious diseases, is the acquired immunodeficiency syndrome (AIDS), which is caused by the infection with the human immunodeficiency virus (HIV).
HIV continues to be a major global public health issue, having claimed more than 35 million lives so far.
In 2017, 940 000 people died from HIV-related causes globally.
There were approximately 36.9 million people living with HIV, at the end of 2017,
with 1.8 million people becoming newly infected in 2017 globally \cite{WHO}. HIV is spread all over the world, however the World Health Organization
African Region is the most affected region, with 25.7 million people living with HIV in 2017.
The African region also accounts for over two thirds of the global total of new HIV infections \cite{WHO}.
In this context,
numerous mathematical models have been proposed for HIV/AIDS transmission and other infectious diseases, but although they
can incorporate important informations about the characteristics of epidemic outbreaks, many of them do not take into account the heterogeneity of the geographical landscape
(see e.g. \cite{khanna2014what, kimbir2012amathematical,rahman2016impact, podder2011mathematical, rocha2018effect, silva2017EcoComplexity}).
However, this geographical heterogeneity, which seems to represent a key factor in understanding the spreading of infectious diseases,
can be studied through the \emph{complex networks} approach, which combines dynamical systems with graph theory,
in order to propose refined mathematical models.

Usually, a complex network is built by considering a graph, given by a set of vertices and a set of edges,
and by coupling each vertex with an instance of a given dynamical system, which can be determined by a set of differential equations.
Recent works have been devoted to complex networks for studying real-world applications, like neural networks
\cite{cao2006global, chen2004global,epperlein2013phase, zhang2009multiple},
behavioral networks \cite{cantin2018nonidentical,cantin2017control}
or biological networks \cite{lassig2001shape,li2010global}.
Other works have also been dedicated to the study of epidemiological problems and their relationship with social networks
\cite{keeling2005networks, may2001infection, moreno2002epidemic}.
The effect of the topology of the graph, which is determined by the disposal of its edges,
on the dynamics of the network and the possible synchronization state are widely studied
\cite{arenas2008synchronization, aziz2006synchronization, belykh2005synchronization,golubitsky2006nonlinear},
but it is still unknown for complex networks of nonidentical systems if one can find some topology that favors a particular dynamics.

In this paper, we propose an original model of complex network of nonidentical dynamical systems, in order to analyze the spread of epidemics within an heterogeneous geographical zone.
To this end, we consider a recent HIV/AIDS model, proposed in \cite{silva2017EcoComplexity}, given by a system of ordinary differential equations,
for which a basic reproduction number $R_0$ can be determined.
It has been proved in \cite{silva2017global}
that this refined model admits a \emph{disease-free equilibrium} (DFE),
which is globally asymptotically stable if $R_0 < 1$,
and an \emph{endemic equilibrium} (EE),
which is globally asymptotically stable if $R_0 > 1$.
Thus we construct a complex network by coupling nonidentical instances of the HIV/AIDS system proposed in \cite{silva2017EcoComplexity},
that is, multiple instances of the epidemiological system with distinct parameters values.
We construct the complex network so that it can take into account the situation
where a part of the population is not concerned with the migrations.
Moreover, we assume that the displacements can be different in some places of the network, that migrations are instantaneous,
and that individuals are not subject to an evolution from one compartment to another during migration from one node to another.
Up to our knowledge, these assumptions are a novelty on the mathematical modeling of the spread of HIV/AIDS epidemics.  
We consider a graph $\mathscr{G} = (\mathscr{V},\,\mathscr{E})$,
where the set of vertices $\mathscr{V}$ models the zones of high population concentration,
and the set of edges $\mathscr{E}$ corresponds to human displacements among those separated zones.
We focus on the situation where the set of vertices of the graph $\mathscr{G}$ is split into at least two subsets:
the first subset being coupled with instances of the HIV model for which the basic reproduction number satisfies $R_0 < 1$,
thus admitting a unique equilibrium which is a DFE;
and the second subset being coupled with instances of the HIV model for which $R_0 > 1$,
thus presenting the coexistence of a DFE and an EE.
We address the following questions.
\emph{Does the coupling between the two subsets of vertices create new equilibrium points?}
\emph{Is it possible to eliminate the endemic equilibriums with a suitable disposition of the couplings?}
\emph{If not, is it possible to minimize the propagation of the epidemic in the network by
	searching an optimal topology?}

Those questions are of great interest, and still have not been deeply studied. 
Recently, it has been proved, in the particular
case of a behavioral model \cite{cantin2018nonidentical}, that oriented chains play an important role
for reaching an expected equilibrium.
Here, the equilibrium we aim to favor in the network is the DFE. 
The main goal of this paper is to find a topology for which the DFEs of a given subset of the network drive the whole network to a global DFE.

This paper is organized as follows.
In Section~\ref{sec:statment}, we recall the equations and stability results of the HIV/AIDS model proposed
in \cite{silva2017EcoComplexity} and analyzed in \cite{silva2017global} that is considered in our study,
and we improve the previous model by constructing a novel complex network model.
In Section~\ref{sec:positregion}, we show that the complex network model is well-posed, by proving
the existence of a positively invariant region that guarantees the non-negativity of the solutions,
together with their boundedness and global existence.
In Section~\ref{sec:stabnetwork}, we explore the influence of the coupling on the equilibrium points of the network,
by computing the global basic reproduction number for particular small networks,
and prove under reasonable assumptions that the network admits a unique DFE which is globally asymptotically stable.
A case study is considered in Section~\ref{sec:casestudyCV}, where the HIV/AIDS epidemic in the Cape Verde archipelago is analyzed and relevant numerical simulations are presented. Moreover, the existence of an optimal topology minimizing the level of infection in the network is investigated.
We end our paper with Section~\ref{sec:concl} with conclusions and discussion of possible future works. 

% ----------------------------------------------

% --------------------------------------------

\section{Problem statement}
\label{sec:statment}

In this section, we recall the HIV/AIDS compartmental model, given by a system of differential equations firstly proposed
in \cite{silva2017EcoComplexity}, which describes the transmission dynamics of HIV in a homogeneous population with variable size.
In Subsection~\ref{section:Construction-of-the-complex-network} we explicit the construction of a complex network determined
with nonidentical instances of the previous HIV/AIDS model.

\subsection{HIV/AIDS model}

We consider a human population affected by a HIV/AIDS epidemics, from \cite{silva2017EcoComplexity}.
The total population is divided into four mutually-exclusive compartments:
susceptible individuals ($S$); 
HIV-infected individuals with no clinical symptoms of AIDS
(the virus is living or developing in the individuals but without producing symptoms or only mild ones) 
but able to transmit HIV to other individuals ($I$); 
HIV-infected individuals under ART treatment
(the so called chronic stage) with a viral load remaining low ($C$); 
and HIV-infected individuals with AIDS clinical symptoms ($A$).
The total population at time $t$, denoted by $N(t)$, is given by
\[
N(t) = S(t) + I(t) + C(t) + A(t).
\]
The SICA model is given by a system of four ordinary differential equations that can be written as follows (see \cite{silva2017global}):
\begin{equation}
\begin{cases}
\dot{S}(t) = \Lambda - \beta \left[ I(t) + \eta_C  C(t) + \eta_A  A(t) \right] S(t) - \mu S(t),\\[0.2 cm]
\dot{I}(t) = \beta \left[ I(t) + \eta_C C(t)  
+ \eta_A  A(t) \right] S(t) - \left(\rho + \phi + \mu\right) I(t) 
+ \omega C(t) + \alpha A(t), \\[0.2 cm]
\dot{C}(t) = \phi I(t) - (\omega + \mu)C(t),\\[0.2 cm]
\dot{A}(t) =  \rho \, I(t) - (\alpha + \mu + d) A(t).
\end{cases}
\label{eq:SICA-model}
\end{equation}

The parameters $\Lambda$, $\mu$, $\beta$, $\eta_C$, $\eta_A$, $\phi$, $\rho$, $\alpha$, $\omega$ and $d$ take fixed values and their significance is detailed in Table~\ref{table:parameters}. 
% ------------------------
\begin{table}[!htb]
	\centering
	\caption{Parameters of the HIV/AIDS model \eqref{eq:SICA-model}.}
	\label{table:parameters}
	\begin{tabular}{l l}
		\hline \hline
		{\small{Symbol}} &  {\small{Description}}\\
		\hline
		{\small{$\Lambda$}} & {\small{Recruitment rate}} \\
		{\small{$\mu$}} & {\small{Natural death rate}} \\
		{\small{$\beta$}} & {\small{HIV transmission rate}} \\
		{\small{$\eta_C$}} & {\small{Modification parameter}} \\
		{\small{$\eta_A$}} & {\small{Modification parameter}}\\	
		{\small{$\phi$}} & {\small{HIV treatment rate for $I$ individuals}} \\
		{\small{$\rho$}} & {\small{Default treatment rate for $I$ individuals}}\\
		{\small{$\alpha$}} & {\small{AIDS treatment rate}}\\
		{\small{$\omega$}} & {\small{Default treatment rate for $C$ individuals}}\\
		{\small{$d$}} & {\small{AIDS induced death rate}}\\
		\hline \hline
	\end{tabular}
\end{table}
% ----------------------------------

We recall that system \eqref{eq:SICA-model} admits a disease-free equilibrium (DFE) given by
\begin{equation}
\Sigma_0 = \left( S^0, I^0, C^0, A^0  \right) 
= \left(\frac{\Lambda}{\mu},0, 0,0  \right).
\label{eq:DFE}
\end{equation}
We introduce the basic reproduction number $R_{0}$,
which represents the expected average number of new HIV infections produced by a single HIV-infected 
individual when in contact with a completely susceptible population, given by
\begin{equation}
R_0 = \frac{ S^0 \beta\, \left[ \xi_2  \left( \xi_1 +\rho \eta_A \right) + \eta_C \phi  \xi_1 \right] }
{\mu \left[ \xi_2  \left( \rho + \xi_1\right) +\phi \xi_1 +\rho d \right] +\rho \omega d} 
= \frac{S^0 \mathcal{N}}{\mathcal{D}},
\label{eq:BRN-R0}
\end{equation}
where 
\[
\begin{split}
&\xi_1 = \alpha + \mu + d, \quad 
\xi_2 = \omega + \mu,\\
&\mathcal{N} = \beta \left[  \xi_2  \left( \xi_1 +\rho\, \eta_A \right) + \eta_C \,\phi \, \xi_1 \right],\\
&\mathcal{D} = \mu \left[  \xi_2  \left( \rho + \xi_1\right) +\phi\, \xi_1 +\rho\,d \right] +\rho \omega d.
\end{split}
\]

\begin{theo}{\cite{silva2017global}}
	\label{theo:global-stab-dfe}
	The disease free equilibrium $\Sigma_0$ given by \eqref{eq:DFE} is globally 
	asymptotically stable for $R_0 < 1$. 
\end{theo}

We also recall that model \eqref{eq:SICA-model} 
has a unique endemic equilibrium $\Sigma_+ =(S^*, I^*, C^*, A^*)$, which is globally asymptotically stable, whenever $R_0 > 1$. 

\begin{lemm}{\cite{silva2017global}}
	\label{lem:uni:ee}
	The model \eqref{eq:SICA-model} has a unique endemic equilibrium 
	$\Sigma_+ =(S^*, I^*, C^*, A^*)$ whenever $R_0 > 1$, 
	which is given by
	\begin{equation*}
	\label{eq:EE}
	S^* = \frac{ \mathcal{D}}{ \mathcal{N}}\, , \quad 
	I^* = \frac{\xi_1 \xi_2 (\Lambda \mathcal{N} 
		- \mu \mathcal{D})}{\mathcal{D} \mathcal{N}} \, , \quad
	C^* = \frac{\phi \xi_1 (\Lambda \mathcal{N} 
		-\mu \mathcal{D})}{\mathcal{D} \mathcal{N}} \, , \quad
	A^* = \frac{\rho \xi_2 (\Lambda \mathcal{N} 
		- \mu \mathcal{D})}{\mathcal{D} \mathcal{N}}.
	\end{equation*}	
\end{lemm} 

\begin{theo}{\cite{silva2017global}}
	\label{theo:globstab-ee}
	The endemic equilibrium $\Sigma_+$ given by \eqref{eq:EE}
	is globally asymptotically stable for $R_0 > 1$.
\end{theo}

The HIV/AIDS model \eqref{eq:SICA-model}  can be rewritten in the following way
\begin{equation}
\dot{x} = f(x,\,p), \quad t \geq 0, \quad x \in \R^4, \quad p \in \R^{10},
\label{SICA-system-short-form}
\end{equation}
where $x = (S,\,I,\,C,\,A)^T$,
$p = (\Lambda,\,\beta,\,\eta_C,\,\eta_A,\,\mu,\,\rho,\,\phi,\,\omega,\,\alpha,\,d)$
and $f(x,\,p)$ is defined by:
\begin{equation*}
f(x,\,p) = \left[
\begin{array}{c}
\Lambda - \beta \left[ I(t) + \eta_C C(t) + \eta_A  A(t) \right] S(t) - \mu S(t)\\[1em]
\beta \left[ I(t) + \eta_C C(t)  
+ \eta_A  A(t) \right] S(t) - \left(\rho + \phi + \mu\right) I(t) 
+ \omega C(t) + \alpha A(t) \\[1em]
\phi I(t) - (\omega + \mu)C(t)\\[1em]
\rho \, I(t) - (\alpha + \mu + d) A(t)
\end{array}
\right].
\end{equation*}

Nonidentical instances of system \eqref{SICA-system-short-form} can be coupled with the vertices of a graph
in order to give rise to a complex network, as we are going to see in the coming subsection.

% ----------------------------------

\subsection{Construction of the complex network}
\label{section:Construction-of-the-complex-network}

Let us consider a graph $\mathscr{G} = (\mathscr{V},\,\mathscr{E})$ made of a finite set $\mathscr{V}$ of $n$ vertices,
where $n$ denotes an integer greater than $2$,
and a finite set $\mathscr{E}$ of $m$ edges, where $m$ denotes a positive integer.
This graph models the geographical zone which is affected by the epidemics.
We assume that $\mathscr{V}$ can be split into at least two subsets of vertices $\mathscr{V}_1$
and $\mathscr{V}_2$.
We couple the vertices of $\mathscr{V}_1$ with an instance of system \eqref{SICA-system-short-form} for which $R_0 < 1$,
and  the vertices of $\mathscr{V}_2$ with an instance of system \eqref{SICA-system-short-form} for which $R_0 > 1$.
The complex network is determined by the following non-linear and autonomous differential system:
\begin{equation}
\dot{X} = F(X,\,P) + L H X,
\label{eq:SICA-network}
\end{equation}
where
\[
\begin{split}
&X = (x_1,\,\dots,\,x_n)^T \in \left(\R^{4}\right)^n,\\
&H X = (H x_1,\,\dots,\,H x_n)^T \in \left(\R^{4}\right)^n,\\
&P = (p_1,\,\dots,\,p_n) \in \left(\R^{10}\right)^n,
\end{split}
\]
and $F$ determines the internal dynamic of each vertex:
\begin{equation*}
F(X,\,P) = \big(f(x_1,\,p_1),\,\dots,\,f(x_n,\,p_n)\big)^T.
\end{equation*}
Furthermore, $L$ is the matrix of connectivity, which is defined as follows.
For each edge $(k,\,j) \in \mathscr{E}$, $k \neq j$,
we have $L_{j,k} > 0$.
If $(k,\,j) \notin \mathscr{E}$, $k \neq j$, we set $L_{j,k} = 0$.
The diagonal coefficients satisfy
\begin{equation*}
L_{j,j} = - \sum_{\substack{k = 1\\ k \neq j}}^n L_{k,j}.
\end{equation*}
Finally, $H$ is the matrix of the coupling strengths and it is given by
\begin{equation*}
H = 
\left[
\begin{array}{l l l l}
\varepsilon_S	&	0				&	0				&	0					\\
0				&	\varepsilon_I	&	0				&	0					\\
0				&	0				&	\varepsilon_C	&	0					\\
0				&	0				&	0				&	\varepsilon_A
\end{array}
\right],
\end{equation*}
with non negative coefficients $\varepsilon_S$, $\varepsilon_I$, $\varepsilon_C$ and $\varepsilon_A$.

In this complex network model, we consider that an edge $(k,\,j) \in \mathscr{E}$, $k \neq j$,
models a connection between two vertices $k$ and $j$,
which corresponds to human displacements from vertex $k$ towards vertex $j$.
Moreover, the parameter $\varepsilon_S$ models the rate of susceptible individuals on vertex $k$ which migrate
towards vertex $j$. The parameters $\varepsilon_I$, $\varepsilon_C$ and $\varepsilon_A$ are defined analogously
for the compartments $I$, $C$ and $A$ respectively.
This implies that our model can take into account the situation where a part of the population is not concerned with the migrations.
Additionally, each connection $(k,\,j)$ is weighted by a positive coefficient $L_{j,k}$, which means that
the displacements can be different in some places of the network.
It is worth emphasizing that the migrations are assumed to be instantaneous,
and that individuals are not subject to an evolution from one compartment to another
during migration from one node to another.

Next we explicit the equations which describe the state of vertex $j \in \lbrace 1,\,\dots,\,n \rbrace$:
\begin{equation}
\begin{cases}
\dot{S}_j = \Lambda_j - \beta_j \big( I_j + \eta_{C,j} \, C_j + \eta_{A,j}  A_j \big) S_j
- \mu_j S_j
+ \varepsilon_S \displaystyle\sum_{k=1}^n L_{j,k} S_k,\\[0.2 cm]
\dot{I}_j = \beta_j \big( I_j + \eta_{C,j} \, C_j + \eta_{A,j}  A_j \big) S_j
- \left(\rho_j + \phi_j + \mu_j\right) I_j 
+ \omega_j C_j + \alpha_j A_j
+ \varepsilon_I \displaystyle\sum_{k=1}^n L_{j,k} I_k, \\[0.2 cm]
\dot{C}_j = \phi_j I_j - (\omega_j + \mu_j)C_j
+ \varepsilon_C \displaystyle\sum_{k=1}^n L_{j,k} C_k,\\[0.2 cm]
\dot{A}_j = \rho_j \, I_j - (\alpha_j + \mu_j + d_j) A_j
+ \varepsilon_A \displaystyle\sum_{k=1}^n L_{j,k} A_k,
\end{cases}
\label{eq:vertex-j-state-equations}
\end{equation}
where the time dependence is omitted, in order to lighten the notations.
The coupling terms can be divided into fluxes exiting from vertex $j$
and fluxes entering in vertex $j$,
that is
\[
\begin{split}
&\displaystyle\sum_{k=1}^n L_{j,k} S_k =
- \left(\sum_{\substack{k = 1\\ k \neq j}}^n L_{k,j} \right) S_j
+ \sum_{\substack{k = 1\\ k \neq j}}^n L_{k,j}  S_k,\quad
\displaystyle\sum_{k=1}^n L_{j,k} I_k =
- \left(\sum_{\substack{k = 1\\ k \neq j}}^n L_{k,j} \right) I_j
+ \sum_{\substack{k = 1\\ k \neq j}}^n L_{k,j}  I_k,\\
&\displaystyle\sum_{k=1}^n L_{j,k} C_k =
- \left(\sum_{\substack{k = 1\\ k \neq j}}^n L_{k,j} \right) C_j
+ \sum_{\substack{k = 1\\ k \neq j}}^n L_{k,j}  C_k,\quad
\displaystyle\sum_{k=1}^n L_{j,k} A_k =
- \left(\sum_{\substack{k = 1\\ k \neq j}}^n L_{k,j} \right) A_j
+ \sum_{\substack{k = 1\\ k \neq j}}^n L_{k,j}  A_k,
\end{split}
\]
for all $j \in \lbrace 1,\,\dots,\,n \rbrace$.

% ------------------------------

\section{Positively invariant region}
\label{sec:positregion}

In this section, we prove that the complex network problem \eqref{eq:SICA-network} is well posed,
and admits a positively invariant region.
This property is obtained as a consequence of the conservation of the couplings terms in system \eqref{eq:SICA-network},
which corresponds to the fact that the matrix of connectivity $L$ is a zero column sum matrix.

\subsection{Preliminary results}

We begin recalling two preliminary results. Since their proofs are well-known, we omit it.

\begin{lemm}
	Let us consider the Cauchy problem
	\begin{equation}
	\begin{cases}
	\dot{\zeta}(t) = f(t) \zeta(t) + g(t), \quad t > t_0,\\
	\zeta(t_0) = \zeta_0,
	\end{cases}
	\label{eq:Cauchy-problem}
	\end{equation}
	where $f$ and $g$ are two continuous functions defined on $\R$, and $t_0 \in \R$.
	We assume that $g(t) \geq 0$ for all $t \in \R$, and that $\zeta(t_0) \geq 0$.
	Let $\zeta(t)$ be a solution of \eqref{eq:Cauchy-problem} defined on $[t_0,\,t_0+\tau]$
	with $\tau > 0$, such that $\zeta(t_0) \geq 0$.
	
	Then we have $\zeta(t) \geq 0$, for all $t \in [t_0,\,t_0+\tau]$.
	\label{lemma:differential-equation}
\end{lemm}

\begin{lemm}
	Let $\zeta$ be a continuous function defined on $[0,\,T]$, with $T > 0$,
	continuously differentiable on $]0,\,T]$,
	and satisfying the differential inequality
	\[
	\dot{\zeta}(t) + \delta_1 \zeta(t) \leq \delta_2, \quad 0 < t \leq T, 
	\]
	with two positive coefficients $\delta_1$, $\delta_2$.
	
	Then we have
	\[
	\zeta(t) \leq \left[\zeta(0) - \frac{\delta_2}{\delta_1}\right] e^{-\delta_1 t} + \frac{\delta_2}{\delta_1},
	\quad 0 \leq t \leq T.
	\]
	\label{lemma:differential-inequality}
\end{lemm}

% -----------------

\subsection{Non-negativity of the solutions of the complex network problem}

The next theorem guarantees the non-negativity of the solutions of the complex network problem \eqref{eq:SICA-network},
which is an obvious property to be satisfied for population dynamics models.
Since the proof uses classical techniques \cite{hethcote2000mathematics, lakshmikantham1989stability},
we only give the main steps.

\begin{theo}
	For any initial condition $X_0 \in (\R^+)^{4n}$, the Cauchy problem
	\begin{equation}
	\begin{cases}
	X(t) = F(X,\,P) + LHX, \quad t > 0,\\
	X(0) = X_0,
	\end{cases}
	\label{eq:Cauchy-problem-complex-network}
	\end{equation}
	where $F$, $P$, $L$ and $H$ are defined as above (see Section~\ref{section:Construction-of-the-complex-network}),
	admits a unique solution defined on $[0,\,T]$ with $T > 0$, whose components are non-negative on $[0,\,T]$.
	\label{theo:non-negativity-solutions}
\end{theo}

\begin{proof}
	Let us consider an initial condition $X_0 \in (\R^+)^{4n}$.
	We denote by $X(t,\,X_0)$ the solution of the Cauchy problem \eqref{eq:Cauchy-problem-complex-network},
	defined on $[0,\,T]$ with $T > 0$.
	
%	We introduce an auxiliary problem defined by
%	\begin{equation}
%	\begin{cases}
%	\dot{S}_j =& \Lambda_j - \beta_j \big( I_j + \eta_{C,j} C_j + \eta_{A,j}  A_j \big) S_j
%	- \mu_j S_j
%	- \varepsilon_S \gamma_j S_j
%	+ \varepsilon_S \displaystyle\sum_{\substack{k = 1\\ k \neq j}}^n L_{k,j}  \abs{S_k},\\[0.2 cm]
%	\dot{I}_j =& \beta_j \big( \abs{I_j} + \eta_{C,j}  \abs{C_j} + \eta_{A,j}  \abs{A_j} \big) S_j
%	- \left(\rho_j + \phi_j + \mu_j\right) I_j 
%	+ \omega_j \abs{C_j} + \alpha_j \abs{A_j}\\[0.2 cm]
%	&- \varepsilon_I \gamma_j I_j
%	+ \varepsilon_I \displaystyle\sum_{\substack{k = 1\\ k \neq j}}^n L_{k,j}  \abs{I_k}, \\[0.2 cm]
%	\dot{C}_j =& \phi_j I_j - (\omega_j + \mu_j)C_j
%	- \varepsilon_C \gamma_j C_j
%	+ \varepsilon_C \displaystyle\sum_{\substack{k = 1\\ k \neq j}}^n L_{k,j}  \abs{C_k},\\[0.2 cm]
%	\dot{A}_j =& \rho_j \, I_j - (\alpha_j + \mu_j + d_j) A_j
%	- \varepsilon_A \gamma_j A_j
%	+ \varepsilon_A \displaystyle\sum_{\substack{k = 1\\ k \neq j}}^n L_{k,j}  \abs{A_k},
%	\end{cases}
%	\label{eq:auxiliary-problem}
%	\end{equation}
	for each $j \in \lbrace 1,\,\dots,\,n \rbrace$,
	where the coefficient $\gamma_j$ corresponds to the fluxes exiting from node $j$,
	and is given by
	\[
	\gamma_j = \displaystyle\sum_{\substack{k = 1\\ k \neq j}}^n L_{k,j}.
	\]
	Let us denote by $\tilde{X}(t,\,X_0)$ the solution of the auxiliary problem \eqref{eq:auxiliary-problem},
	stemming from the same initial condition $X_0 \in (\R^+)^{4n}$,
	defined on $[0,\,\tilde{T}]$.
	
	Applying Lemma~\ref{lemma:differential-equation},
	we easily prove that the components of $\tilde{X}(t,\,X_0)$ are non-negative.
	This implies that $\tilde{X}(t,\,X_0)$ is also a solution of the Cauchy problem \eqref{eq:Cauchy-problem-complex-network}
	on $[0,\,\tilde{T}]$.
	By uniqueness, we have $\tilde{X}(t,\,X_0) = X(t,\,X_0)$ for all $t \in [0,\,T] \cap [0,\,\tilde{T}]$.
	Finally, it is seen that $T = \tilde{T}$, which achieves the proof.
\end{proof}

\subsection{Boundedness of the solutions of the complex network problem}

Let us introduce the minimum mortality rate $\mu_0$ defined by
\[
\mu_0 = \min\limits_{1 \leq j \leq n} \mu_j,
\]
the positive coefficient $\Lambda_0$ defined by
\[
\Lambda_0 = \displaystyle\sum_{j=1}^n \Lambda_j,
\]
and the compact region
\begin{equation}
\Omega = \left\lbrace (x_j)_{1\leq j \leq 4n} \in (\R^+)^{4n}~;~\displaystyle\sum_{j=1}^{4n} x_j \leq \frac{\Lambda_0}{\mu_0} \right\rbrace.
\label{eq:region-Omega}
\end{equation}
The total population in the complex network, defined by
\[
N(t) = \displaystyle\sum_{j=1}^n \big[S_j(t) + I_j(t) + C_j(t) + A_j(t) \big], \quad t \in [0,\,T],
\]
satisfies
\[
\dot{N}(t) \leq - \mu_0 N(t) + \Lambda_0,
\quad t \in [0,\,T],
\]
since the matrix of connectivity $L$ is a zero column sum matrix.
Applying Lemma~\ref{lemma:differential-inequality} leads to
\[
N(t) \leq  \left[N(0) - \frac{\Lambda_0}{\mu_0}\right]e^{-\mu_0 t} + \frac{\Lambda_0}{\mu_0}, \quad t \in [0,\,T],
\]
thus we obtain the following theorem.

\begin{theo}
	The region $\Omega$ defined by \eqref{eq:region-Omega} is positively invariant under the flow induced
	by the complex network \eqref{eq:SICA-network}.
	\label{theo:boundedness-solutions}
\end{theo}

\begin{rema}
	It is easily seen that the positively invariant region $\Omega$ for the complex network problem \eqref{eq:SICA-network}
	satisfies
	\[
	\prod_{i=1}^n \Omega_i \subset \Omega,
	\]
	where $\Omega_i = \left\lbrace (x_j)_{1\leq j \leq 4} \in (\R^+)^{4}~;~\displaystyle\sum_{j=1}^{4} x_j \leq \frac{\Lambda_i}{\mu_i} \right\rbrace$
	corresponds to the positively invariant region of the node $(i)$ in absence of coupling.
	Roughly speaking, the couplings can enlarge the phase space of the flow induced by the network problem.
\end{rema}

\section{Stability analysis of the complex network}
\label{sec:stabnetwork}

In this section, we explore the effect of the couplings on the dynamics of the complex network \eqref{eq:SICA-network}.
We use symbolic computational methods, in the case of small networks.
Furthermore, we prove the existence of a unique disease-free equilibrium which is globally asymptotically stable.

\subsection{Asymmetric two-nodes network}

Let us consider a two-nodes networks,
with one vertex $(1)$ for which $R_0 < 1$,
another vertex $(2)$ for which $R_0 > 1$,
and a directed connection from vertex $(1)$ towards vertex $(2)$ (see Figure~\ref{fig:two-nodes-network}).
In that case, the matrix of connectivity is given by
\[
L =
\left[
\begin{array}{l l}
- L_{2,1}		&		0		\\
+ L_{2,1}		&		0		
\end{array}
\right],
\]
and therefore the equations of the network read
\begin{equation}
\begin{cases}
\dot{S}_1 = \Lambda_1 - \beta_1 \left( I_1 + \eta_{C,1} \, C_1 + \eta_{A,1}  A_1 \right) S_1 - \mu_1 S_1
- L_{2,1} \varepsilon_S S_1,\\[0.2 cm]
\dot{I}_1 = \beta_1 \left( I_1 + \eta_{C,1} \, C_1  
+ \eta_{A,1}  A_1 \right) S _1- \left(\rho_1 + \phi_1 + \mu_1\right) I _1
+ \omega C_1 + \alpha A_1
- L_{2,1} \varepsilon_I I_1, \\[0.2 cm]
\dot{C}_1 = \phi_1 I_1 - (\omega_1 + \mu_1)C_1
- L_{2,1} \varepsilon_C C_1,\\[0.2 cm]
\dot{A}_1 =  \rho_1 \, I_1 - (\alpha_1 + \mu_1 + d_1) A_1
- L_{2,1} \varepsilon_A A_1, \\[0.3cm]
\dot{S}_2 = \Lambda_2 - \beta_2 \left( I_2 + \eta_{C,2} \, C_2 + \eta_{A,2}  A_2 \right) S_2 - \mu_2 S_2
+ L_{2,1} \varepsilon_S S_1,\\[0.2 cm]
\dot{I}_2 = \beta_2 \left( I_2 + \eta_{C,2} \, C_2  
+ \eta_{A,2}  A_2 \right) S_2 - \left(\rho_2 + \phi_2 + \mu_2\right) I_2 
+ \omega_2 C_2 + \alpha_2 A_2
+ L_{2,1} \varepsilon_I I_1, \\[0.2 cm]
\dot{C}_2 = \phi_2 I_2 - (\omega_2 + \mu_2)C_2
+ L_{2,1} \varepsilon_C C_1,\\[0.2 cm]
\dot{A}_2 =  \rho_2 \, I_2 - (\alpha_2 + \mu_2 + d_2) A_2
+ L_{2,1} \varepsilon_A A_1,
\end{cases}
\label{eq:two-nodes-network}
\end{equation}
where we omit the dependence in $t$ in order to lighten our notations.

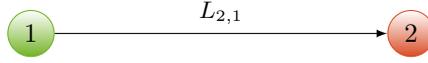
\begin{figure}[h!]
	\centering
	\tikzstyle{greencircle} = [circle,
	draw,
	green!40!brown,
	bottom color=green!40!brown,
	top color= white,
	text=black]
	\tikzstyle{redcircle} = [circle,
	draw,
	red!40!brown,
	bottom color=red!40!brown,
	top color= white,	
	text=black]
	\begin{tikzpicture}
	\node[greencircle] (x1) at (-2.5,0) {$1$};
	\node[redcircle] (x2) at (2.5,0) {$2$};
	\draw[>=latex,->] (x1) to  node [pos=0.5,above]{\small $L_{2,1}$}(x2);
	\end{tikzpicture}
	\caption{Asymmetric two-nodes network, built with two nonidentical instances of system \eqref{eq:SICA-model}.
		The green node (1) is associated with an instance of system \eqref{eq:SICA-model} for which the basic reproduction number $R_0$ satisfies
		$R_0 < 1$, whereas the red node (2) is coupled with an instance of system \eqref{eq:SICA-model} for which $R_0 > 1$.}
	\label{fig:two-nodes-network}
\end{figure}

Roughly speaking, the coupling coefficient $L_{2,1}$ acts on vertex $(1)$ as if the mortality rate $\mu_1$ increases,
which changes the value of the basic reproduction number on vertex $(1)$.
Thus the following question arises.
\emph{How does $R_0$ vary when $L_{2,1}$ increases?}
Proposition~\ref{prop:R0:2nodes} below partly answers this question. 

First, we easily prove that system \eqref{eq:two-nodes-network} admits a disease-free equilibrium point $\Sigma^0$
given by
\begin{equation}
\label{eq:DFE:2nodes}
\Sigma^0 = (S_1^0, I_1^0, C_1^0, A_1^0, S_2^0, I_2^0, C_2^0, A_2^0)
= \left( \frac{ \Lambda_1}{ L_{2,1}   \varepsilon_S +\mu_1}, 0, 0, 0,
\frac{ L_{2,1} \varepsilon_S ( \Lambda_1+ \Lambda_2 )
	+ \Lambda_2 \mu_1}{\mu_2\, \left( L_{2,1} \varepsilon_S +\mu_1 \right) }, 0, 0, 0 \right) \, .
\end{equation}

Following the method proposed in \cite{driessche2002reproduction} for computing the basic reproduction number, we have the following result. 

\begin{prop}
	\label{prop:R0:2nodes}
	The basic reproduction number of model \eqref{eq:two-nodes-network} is
	\begin{equation}
	\label{eq:R0:2nodes}
	R_0 = \max \{ R_{0,1}, R_{0,2} \} \, 
	\end{equation}
	where
	\begin{equation}
	\label{eq:R01:2nodes}
	R_{0,1} = \frac{\mathcal{N}_1}{\mathcal{D}_1}
	\end{equation}
	and
	\begin{equation}
	\label{eq:R02:2nodes}
	R_{0,2} = \frac{\mathcal{N}_2 (\Lambda_1  L_{2,1} \varepsilon_S + \Lambda_2 L_{2,1} \varepsilon_S + \Lambda_2 \mu_1)}
	{\mathcal{D}_2 \left[(L_{2,1} \varepsilon_S+\mu_1 ) \mu_2 \right]}
	\end{equation}
	with
	\begin{equation*}
	\begin{split}
	&\mathcal{N}_2 = \beta_2 \left[  \xi_{2,2}  \left( \xi_{1,2} +\rho_2\, \eta_{A, 2} \right) + \eta_{C, 2} \,\phi_2 \, \xi_{1, 2} \right],\\ 
	&\mathcal{D}_2 = \mu_2 \left[  \xi_{2, 2}  \left( \rho_2 + \xi_{1, 2}\right) +\phi_2\, \xi_{1,2} +\rho_2 \,d_2 \right] +\rho_2 \omega_2 d_2,
	\end{split}
	\end{equation*}
	and
	\begin{equation*}
	\begin{split}
	\mathcal{N}_1 &= \Lambda_1 \beta_1 \left(L_{2,1}^{2} \varepsilon_A \varepsilon_C
	+ \left( \eta_{A,1} \rho_1 \varepsilon_C + \eta_{C,1} \phi_1 \varepsilon_A
	+ \xi_{1,1} \varepsilon_C + \xi_{2,1} \varepsilon_A  \right) L_{2,1} 
	+ \xi_{2,1} (\xi_{1,1} + \rho_1 \eta_{A,1} ) +\xi_{1,1} \eta_{C,1} \phi_1 \right),\\
	\mathcal{D}_1 &= (L_{2,1} \varepsilon_S + \mu_1)
	\left( L_{2,1}^3 \varepsilon_A \varepsilon_C \varepsilon_I +  p_2 L_{2,1}^2 +  p_1 L_{2,1} +  p_0  \right),\\
	p_0 &= \mu_1 \big[\xi_{2,1} (\rho_1+\xi_{1,1})+\phi_1 \xi_{1,1} + d_1 \rho_1\big]+ d_1 \omega_1 \rho_1 ,\\
	p_1 &= \big[\xi_{1,1} (\mu_1 + \phi_1) + (d_1 + \mu_1) \rho_1  \big]\varepsilon_C
	+ \xi_{1,1} \xi_{2,1}  \varepsilon_I
	+ \big[ (\mu_1  +  \rho_1)\xi_{2,1}  + \mu_1 \phi_1  \big]\varepsilon_A,  \\
	p_2 &= \big[(\rho_1+\phi_1+\mu_1)\varepsilon_A+\varepsilon_I \xi_{1,1}\big]\varepsilon_C+\varepsilon_A \varepsilon_I \xi_{2,1}.
	\end{split}
	\end{equation*}
\end{prop}

\begin{proof}
	
	Let $\mathcal{F}_i(t)$ be the rate at which new infections appear
	in the $i$-th compartment and
	$\mathcal{V}_i^+(t)$ be the ``individuals'' transfer rate 
	into the $i$-th compartment in all other ways. Similarly, let
	$\mathcal{V}_i^-(t)$ denote the ``individuals'' transfer rate 
	out of the $i$-th compartment, for which
	\[
	\left[
	\dot{S}_1(t), \dot{I}_1(t), \dot{C}_1(t), \dot{A}_1(t), \dot{S}_2(t), \dot{I}_2(t), \dot{C}_2(t), \dot{A}_2(t)
	\right]^T
	=\mathcal{F}(t)-\mathcal{V}(t)=\mathcal{F}(t)-\big(\mathcal{V}^-(t)-\mathcal{V}^+(t)\big).
	\]	
	
	Therefore, we take
	\begin{align*}
	\mathcal{F}(t)=\left[
	\begin{matrix}
	0\\
	\displaystyle S_1 \beta_1 (\eta_{A,1} A_1 + \eta_{C,1} C_1  + I_1)\\
	0\\
	0\\
	0\\
	\displaystyle L_{2,1} \varepsilon_I I_1+S_2 \beta_2 (\eta_{A,2} A_2 + \eta_{C,2} C_2 + I_2)\\
	0\\
	0
	\end{matrix}
	\right],
	\quad \mathcal{V}^+(t)=\left[
	\begin{matrix}
	\displaystyle\Lambda_1 \\
	\displaystyle\omega_1 C_1 + \alpha_1 A_1 \\
	\displaystyle\phi_1 I_1 \\
	\displaystyle\rho_1 \, I_1 \\
	\displaystyle\Lambda_2 + L_{2,1} \varepsilon_S S_1\\
	\displaystyle\omega_2 C_2 + \alpha_2 A_2  \\
	\displaystyle\phi_2 I_2 + L_{2,1} \varepsilon_C C_1 \\
	\displaystyle\rho_2 \, I_2 + L_{2,1} \varepsilon_A A_1
	\end{matrix}
	\right]
	\end{align*}
	and
	\begin{align*}
	\mathcal{V}^-(t)=\left[
	\begin{matrix}
	\displaystyle \beta_1 \left( I_1 + \eta_{C,1} \, C_1 + \eta_{A,1}  A_1 \right) S_1 + \mu_1 S_1 + L_{2,1} \varepsilon_S S_1 \\
	\displaystyle \left(\rho_1 + \phi_1 + \mu_1\right) I _1 + L_{2,1} \varepsilon_I I_1\\
	\displaystyle(\omega_1 + \mu_1)C_1 + L_{2,1} \varepsilon_C C_1 \\
	\displaystyle(\alpha_1 + \mu_1 + d_1) A_1 + L_{2,1} \varepsilon_A A_1 \\
	\displaystyle \beta_2 \left( I_2 + \eta_{C,2} \, C_2 + \eta_{A,2}  A_2 \right) S_2 + \mu_2 S_2 \\
	\displaystyle\left(\rho_2 + \phi_2 + \mu_2\right) I_2 \\
	\displaystyle(\omega_2 + \mu_2)C_2 \\
	\displaystyle(\alpha_2 + \mu_2 + d_2) A_2
	\end{matrix}\right]  \, .
	\end{align*}
	
	The Jacobian matrices $F$ of $\mathcal{F}(t)$ and $V$ of $\mathcal{V}(t)$ are given by
	\[
	F= \big[F^1,\,F^2\big], \quad V = \big[V^1,\,V^2\big],
	\]
	where
	\begin{equation*}
	\begin{split}
	&F^1 = 
	\left[
	\begin {array}{cccc}
	0 	& 0 	& 0 	& 0 																									\\ 
	\beta_1  \left( \eta_{A,1} A_1 + \eta_{C,1} C_1 + I_1\right) & S_1 \beta_1& S_1 \beta_1 \eta_{C,1} & S_1 \beta_1 \eta_{A,1} \\ 
	0 	& 0 	& 0 	& 0 \\ 
	0 	& 0 	& 0 	& 0 \\ 
	0 	& 0 	& 0 	& 0 \\ 
	0	&L_{2,1} \varepsilon_I & 0 & 0 \\ 
	0 	& 0 	& 0 	& 0 \\ 
	0 	& 0 	& 0 	& 0 
	\end{array}
	\right], \\
	&F^2 =
	\left[
	\begin{array}{cccc}
	0 & 0 & 0 & 0\\ 
	0 & 0 & 0 & 0\\ 
	0 & 0 & 0 & 0\\ 
	0 & 0 & 0 & 0\\ 
	0 & 0 & 0 & 0\\ 
	\beta_2 \left( \eta_{A,2} A_2  + \eta_{C,2} C_2  + I_2  \right) & S_2 \beta_2 & S_2 \beta_2 \eta_{C,2} & S_2 \beta_2 \eta_{A,2}\\ 
	0 & 0 & 0 & 0\\ 
	0 & 0 & 0 & 0
	\end{array}
	\right], \\
	&V^1 =
	\left[
	\begin{array}{cccc} 
	\beta_1 (A_1 \eta_{A,1} + C_1 \eta_{C,1} + I_1) + L_{2,1} \varepsilon_S +\mu_1 & S_1 \beta_1 & S_1 \beta_1 \eta_{C,1} & S_1 \beta_1 \eta_{A,1} \\ 
	0	& L_{2,1} \varepsilon_I +\mu_1+ \phi_1 + \rho_1 & -\omega_1 & -\alpha_1 \\ 
	0	& -\phi_1 & L_{2,1} \varepsilon_C + \xi_{2,1} & 0 \\
	0	& -\rho_1 & 0 & L_{2,1} \varepsilon_A + \xi_{1,1} \\ 
	- L_{2,1} \varepsilon_S & 0 & 0 & 0 \\ 
	0 	& 0 & 0 & 0 \\ 
	0 	& 0 &-L_{2,1} \varepsilon_C & 0\\ 
	0 	& 0 & 0 &-L_{2,1} \varepsilon_A
	\end{array}
	\right], \\
	&V^2 =
	\left[
	\begin{array}{cccc} 
	0 & 0 & 0 & 0\\ 
	0 & 0 & 0 & 0 \\ 
	0 & 0 & 0 & 0\\
	0 & 0 & 0 & 0\\ 
	\beta_2 (A_2 \eta_{A,2} +  C_2 \eta_{C,2} + I_2) + \mu_2& S_2 \beta_2 & S_2  \beta_2 \eta_{C,2} & S_2 \beta_2 \eta_{A,2}\\ 
	0& \rho_2 + \phi_2 + \mu_2 & -\omega_2 & -\alpha_2 \\ 
	0&-\phi_2& \xi_{2,2} & 0\\ 
	0&-\rho_2 & 0& \xi_{1,2}
	\end {array}
	\right],
	\end{split}
	\end{equation*}
	and
	\begin{equation*}
	\begin{split}
	&\xi_{1,1} = \alpha_1 + \mu_1 + d_1, \quad 
	\xi_{2,1} = \omega_1 + \mu_1,\\
	&\xi_{1,2} = \alpha_2 + \mu_2 + d_2, \quad 
	\xi_{2,2} = \omega_2 + \mu_2 \, .
	\end{split}
	\end{equation*}
	Evaluating the matrices $F$ and $V$ at the disease-free equilibrium $\Sigma^0$ given by \eqref{eq:DFE:2nodes},
	we find
	\[
	F_0 = \big[ F_0^1,\,F_0^2 \big], \quad
	V_0 = \big[ V_0^1,\,V_0^2 \big],
	\]
	with
	\begin{equation*}
	\begin{split}
	&F_0^1=
	\left[
	\begin{array}{cccc} 
	0 & 0 & 0 & 0\\ 
	0 & \frac{\Lambda_1  \beta_1}{ L_{2,1} \varepsilon_S +\mu_1}&
	\frac{ \Lambda_1  \beta_1 \eta_{C,1} }{ L_{2,1} \varepsilon_S + \mu_1}&
	\frac{\Lambda_1 \beta_1 \eta_{A,1}}{ L_{2,1} \varepsilon_S +\mu_1}\\ 
	0 & 0 & 0 & 0\\ 
	0 & 0 & 0 & 0\\ 
	0 & 0 & 0 & 0\\ 
	0 &  L_{2,1} \varepsilon_I  & 0 & 0\\ 
	0 & 0 & 0 & 0\\ 
	0 & 0 & 0 & 0
	\end{array}
	\right], \\
	&F_0^2 =
	\left[
	\begin{array}{cccc} 
	0 & 0 & 0 & 0\\ 
	0 & 0 & 0 & 0\\ 
	0 & 0 & 0 & 0\\ 
	0 & 0 & 0 & 0\\ 
	0 & 0 & 0 & 0\\ 
	0&\frac { \left( \Lambda_1 L_{2,1} \varepsilon_S +\Lambda_2 L_{2,1} \varepsilon_S + \Lambda_2 \mu_1 \right) \beta_2}{\mu_2  \left( L_{2,1} \varepsilon_S +\mu_1 \right) } & \frac { \left( \Lambda_1 L_{2,1} \varepsilon_S +\Lambda_2 L_{2,1} \varepsilon_S +\Lambda_2 \mu_1 \right) \beta_2 \eta_{C,2}}{\mu_2
		\left( L_{2,1} \varepsilon_S +\mu_1 \right) }& \frac { \left( 
		\Lambda_1 L_{2,1} \varepsilon_S +\Lambda_2 L_{2,1}  \varepsilon_S + \Lambda_2  \mu_1 \right) \beta_2 \eta_{A,2}}{\mu_2 \left( L_{2,1}
		\varepsilon_S + \mu_1 \right) }\\ 
	0 & 0 & 0 & 0\\ 
	0 & 0 & 0 & 0 
	\end{array}
	\right], \\
	&V_0^1=
	\left[ 
	\begin{array}{cccc} 
	L_{2,1} \varepsilon_S + \mu_1 & 
	\frac {\Lambda_1\,\beta_1}{ L_{2,1} \varepsilon_S + \mu_1}&
	\frac { \Lambda_1 \beta_1 \eta_{C,1} }{ L_{2,1} \varepsilon_S +\mu_1}&{
		\frac { \Lambda_1 \beta_1 \eta_{A,1}}{L_{2,1} \varepsilon_S + \mu_1}}\\ 
	0& L_{2,1} \varepsilon_I + \mu_1 + \phi_1 + \rho_1 & -\omega_1 & -\alpha_1\\ 
	0 & -\phi_1 &L_{2,1}
	\varepsilon_C  + \xi_{2,1} & 0\\ 
	0& -\rho_1 & 0 & L_{2,1} \varepsilon_A +\alpha_1 d_1 + \mu_1\\ 
	- L_{2,1} \varepsilon_S & 0 & 0 & 0 \\
	0 & 0 & 0 & 0\\ 
	0 & 0 & -L_{2,1} \varepsilon_C & 0 \\ 
	0 & 0 & 0 &- L_{2,1} \varepsilon_A
	\end{array}
	\right],\\
	&V_0^2=
	\left[ 
	\begin{array}{cccc} 
	0 & 0 & 0 & 0\\ 
	0 & 0 & 0 & 0\\ 
	0 & 0 & 0 & 0\\ 
	0 & 0 & 0 & 0 \\ 
	\mu_2 & 
	\frac { \left( \Lambda_1 L_{2,1} \varepsilon_S + \Lambda_2 L_{2,1} \varepsilon_S + \Lambda_2  \mu_1 \right) \beta_2}{\mu_2 \left( L_{2,1} \varepsilon_S + \mu_1 \right)} & \frac{ \left( \Lambda_1 L_{2,1} \varepsilon_S +\Lambda_2 L_{2,1} \varepsilon_S + \Lambda_2 \mu_1 \right) \beta_2 \eta_{C,2} }{ \mu_2 \left( L_{2,1} \varepsilon_S + \mu_1 \right) } & \frac{ \left( \Lambda_1 L_{2,1} \varepsilon_S +
		\Lambda_2 L_{2,1} \varepsilon_S + \Lambda_2 \mu_1 \right) \beta_2 \eta_{A,2} }{ \mu_2  \left( L_{2,1} \varepsilon_S + \mu_1 \right) }\\
	0 & \rho_2 + \phi_2 + \mu_2 & -\omega_2 & -\alpha_2\\ 
	0 & -\phi_2 & \omega_2 + \mu_2 & 0\\ 
	0 &-\rho_2 & 0 &\alpha_2 + \mu_2 + d_2
	\end{array}
	\right].
	\end{split}
	\end{equation*}
	The eigenvalues of the matrix $F_0 V_0^{-1}$ are given by: 
	\begin{equation*}
	\left[ 
	\begin {array}{c c c c c c c c} 
	0,&
	0,&
	0,&
	0,&
	0,&
	0,&
	\frac{\mathcal{N}_1}{\mathcal{D}_1},&
	\frac{\mathcal{N}_2 (\Lambda_1  L_{2,1} \varepsilon_S + \Lambda_2 L_{2,1} \varepsilon_S + \Lambda_2 \mu_1)}{\mathcal{D}_2 ((L_{2,1} \varepsilon_S+\mu_1) \mu_2 )}
	\end {array}
	\right]^T .
	\end{equation*}
	
	The basic reproduction number is given by the dominant eigenvalue of the matrix $F_0 V_0^{-1}$, that is, $R_0$ takes the value given by \eqref{eq:R0:2nodes}. 
	
\end{proof}

\begin{rema}
	We emphasize that $R_{0,1}$ and $R_{0,2}$ correspond to the basic reproduction numbers
	of nodes $(1)$ and $(2)$ respectively, in absence of coupling (that is $\varepsilon_S = \varepsilon_I = \varepsilon_C = \varepsilon_A = 0$).
	Thus, the expression $R_0 = \max(R_{0,1},\,R_{0,2})$ implies that if $R_{0,1}>1$ or $R_{0,2}>1$, then $R_0>1$.
	In other words, the node admitting a basic reproduction number $R_0>1$ drives the other node to a global Endemic Equilibrium (EE).
	It is a work in progress to generalize this pattern to more general topologies (e.g. chain networks).
	However, one should not conclude for the general case that the good solution is to ``cut'' the couplings
	(there may exist an optimal coupling topology which globally tempers the level of infected individuals,
	as we are going to show in Section~\ref{sec:casestudyCV} below).
\end{rema}

We easily prove that
$R_{0,2}$ is an increasing function of $L_{2,1}$.
Indeed, we have
\[
R_{0,2} = k \dfrac{d_1 L_{2,1} + d_2}{d_3 L_{2,1} + d_4},
\]
with $k = \frac{\mathcal{N}_2}{\mathcal{D}_2 \mu_2}$,
$d_1 = \varepsilon_S (\Lambda_1 + \Lambda_2)$,
$d_2 = \Lambda_2 \mu_1$,
$d_3 = \varepsilon_S$ and $d_4 = \mu_1$.
Since $d_1 d_4 - d_2 d_3 = \mu_1 \varepsilon_S \Lambda_1 > 0$,
we can conclude that $R_{0,2}$ is an increasing function of $L_{2,1}$.
Figure~\ref{fig:influence-coupling-R01-R02}(a) illustrates this increasing shape of $R_{0,2}$ with respect to $L_{2,1}$
for the following parameters values:
\[
\begin{split}
&\Lambda_1 = \Lambda_2 = 2, \beta_1 = 0.0015, \beta_2 = 0.001,
\eta_{C,1} = \eta_{C,2} = 0.04,
\eta_{A,1} = \eta_{A,2} = 1.3,
\mu_1 = \mu_2 = \tfrac{1}{70},\\
&\rho_1 = \rho_2 = 0.1,
\varphi_1 = \varphi_2 = 1, \omega_1 = \omega_2 = 0.09,
\alpha_1 = \alpha_2 = 0.33, d_1 = d_2 = 1,
\varepsilon = 0.1.
\end{split}
\]

At the opposite, one can find parameters values for which $R_{0,1}$ is a decreasing function of $L_{2,1}$,
but also other parameters values for which $R_{0,1}$ is an increasing function of $L_{2,1}$ in a
neighborhood of $0$.
Figure~\ref{fig:influence-coupling-R01-R02}(b) presents an example for which $R_{0,1}$ admits a maximum
with respect to $L_{2,1}$; this example has been obtained
for the following parameters values:
\[
\begin{split}
&\Lambda_1 = \Lambda_2 = 1,
\varepsilon = 1,
\eta_{A,1} = \eta_{C,1} = \eta_{A,2} = \eta_{C,2} = 0,
\beta_1 = \beta_2 = 1,
\mu_1 = 10,\\
&d_1 = 1,
\omega_1 = 1,
\varphi_1 = 1,
\rho_1 = 1,
\alpha_1 = 50.
\end{split}
\]

\begin{figure}[h!]
	\centering
	\includegraphics[scale=0.9]{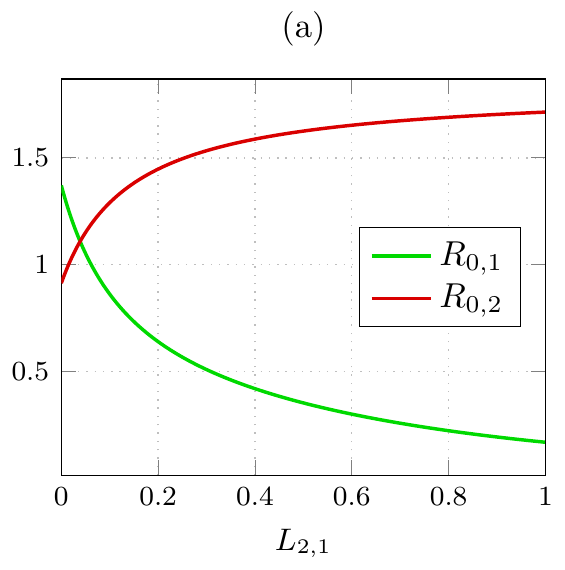}
	\includegraphics[scale=0.9]{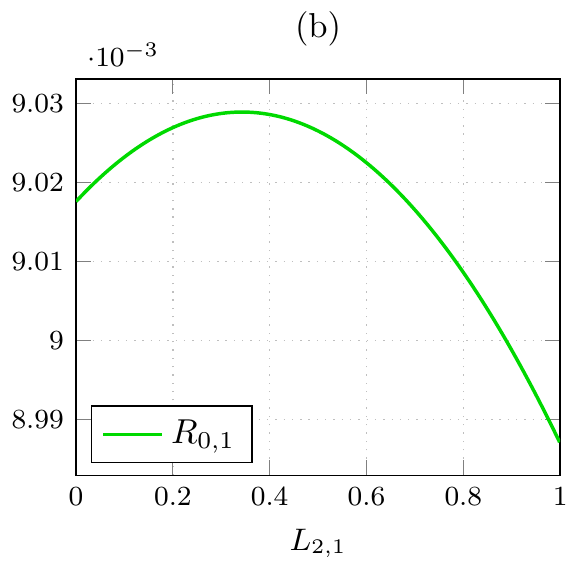}
	\caption{Influence of the coupling on the basic reproduction numbers $R_{0,1}$ and $R_{0,2}$ of system
		\eqref{eq:two-nodes-network} for different parameters values.
		$R_{0,2}$ is an increasing function of the coupling strength $L_{2,1}$ (a),
		whereas $R_{0,1}$ can admit a maximum (b).}
	\label{fig:influence-coupling-R01-R02}
\end{figure}

In parallel, the coupling strengths $\varepsilon_S$, $\varepsilon_I$, $\varepsilon_C$ and
$\varepsilon_A$, stored in matrix $H$ (see section \ref{section:Construction-of-the-complex-network}),
are also observed to play an important role (see Figure \ref{fig:influence-coupling-strengths-R01-R02}).
For the same set of parameters as above, and a frozen coefficient $L_{2,1} = 0.15$, we have
computed the values of the basic reproduction numbers $R_{0,1}$ and $R_{0,2}$ with respect to
a variation of $\varepsilon_S$, $\varepsilon_I$, $\varepsilon_C$ and
$\varepsilon_A$.
It seems that $R_{0,1}$ and $R_{0,2}$ are robust to a variation of the coupling strengths $\varepsilon_I$ and 
$\varepsilon_A$, whereas a variation of $\varepsilon_S$ or $\varepsilon_C$ can induce an important variation
in $R_{0,1}$ and $R_{0,2}$, which can imply a change in the dynamics of both nodes in the complex network \eqref{eq:two-nodes-network}.
Moreover, the coupling strengths $\varepsilon_S$ and $\varepsilon_C$ seem to play antagonistic roles, since
an increase of $\varepsilon_S$ provokes an increase of $R_{0,2}$,
whereas an increase of $\varepsilon_C$ provokes an increase of $R_{0,1}$.

\begin{figure}[h!]
	\centering
	\includegraphics[scale=0.9]{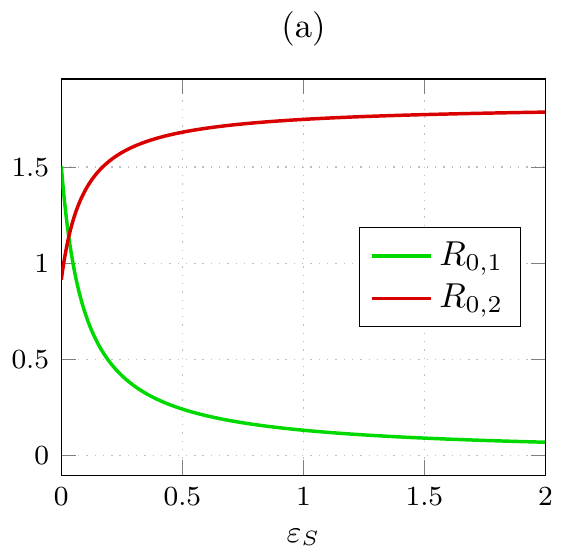}
	\includegraphics[scale=0.9]{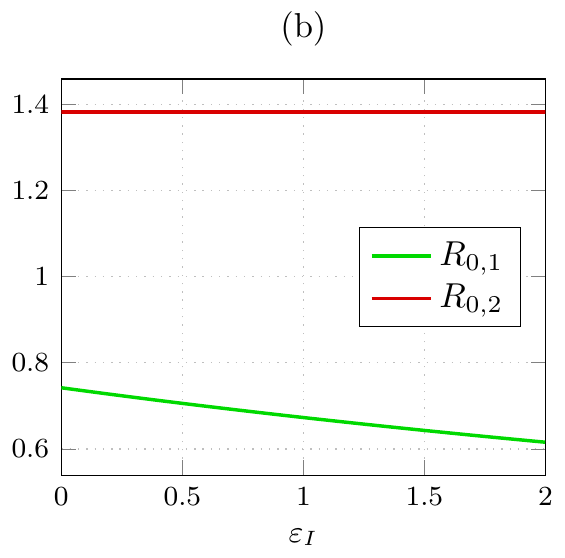}
	
	\includegraphics[scale=0.9]{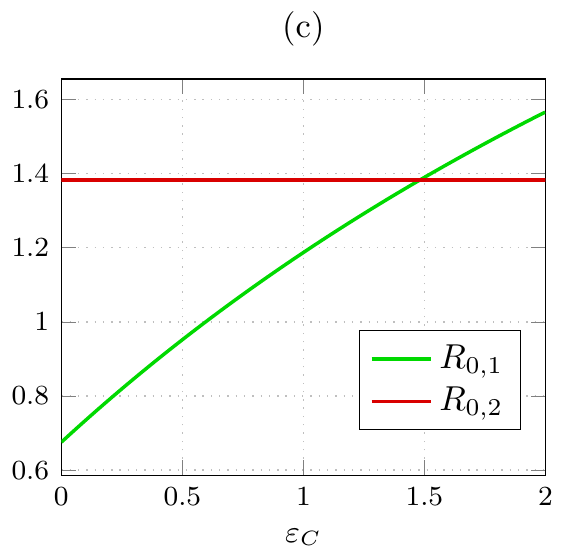}
	\includegraphics[scale=0.9]{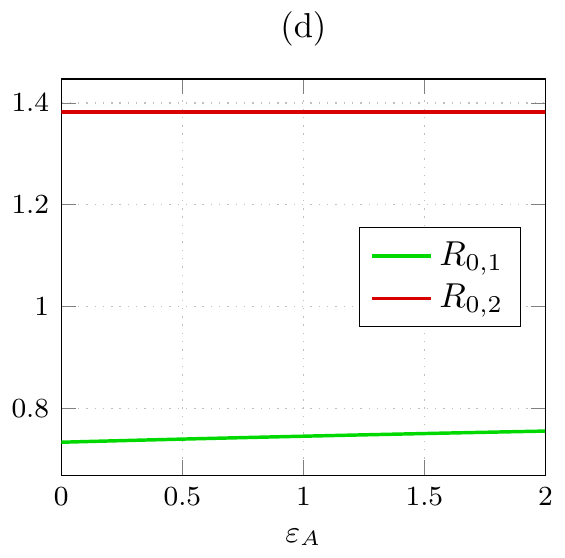}
	\caption{Influence of the coupling strengths $\varepsilon_S$, $\varepsilon_I$,
		$\varepsilon_C$ and $\varepsilon_A$ on the basic reproduction numbers $R_{0,1}$ and $R_{0,2}$ of system
		\eqref{eq:two-nodes-network}.
		A variation of $\varepsilon_S$ or $\varepsilon_C$ can induce a remarkable change in the values of $R_{0,1}$ and $R_{0,2}$ [(a), (c)].
		At the opposite, $R_{0,1}$ and $R_{0,2}$ seem to be robust to a variation of $\varepsilon_I$ or $\varepsilon_A$ [(b), (d)].}
	\label{fig:influence-coupling-strengths-R01-R02}
\end{figure}

\begin{rema}
	After tedious symbolical computations, it is possible to obtain the expressions of the basic reproduction numbers in the case
	of a symmetric two-nodes network (see Figure~\ref{fig:symmetric-two-nodes-network}).
	However, the output is unreadable, even with relevant simplifications of the parameters of the system.
	
	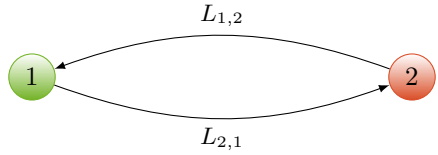
\begin{figure}[h!]
		\centering
		\tikzstyle{greencircle} = [circle,
		draw,
		green!40!brown,
		bottom color=green!40!brown,
		top color= white,
		text=black]
		\tikzstyle{redcircle} = [circle,
		draw,
		red!40!brown,
		bottom color=red!40!brown,
		top color= white,	
		text=black]
		\begin{tikzpicture}
		\node[greencircle] (x1) at (-2.5,0) {$1$};
		\node[redcircle] (x2) at (2.5,0) {$2$};
		\draw[>=latex,->, bend right=20] (x1) to  node [pos=0.5,below]{\small $L_{2,1}$}(x2);
		\draw[>=latex,->, bend right=20] (x2) to  node [pos=0.5,above]{\small $L_{1,2}$}(x1);
		\end{tikzpicture}
		\caption{Symmetric two-nodes network.}
		\label{fig:symmetric-two-nodes-network}
	\end{figure}
	
	In the same manner, it is possible to obtain the complete symbolical expression of the disease-free equilibrium $\Sigma^0$
	of a three-nodes chain (see Figure \ref{fig:three-nodes-chain}):
	\begin{equation}
	\label{eq:DFE:3nodes}
	\Sigma^0 = \big(S^1_0,\, I^1_0,\, C^1_0,\, A^1_0,\, S^2_0,\, I^2_0,\, C^2_0,\, A^2_0,\, S^3_0,\, I^3_0,\, C^3_0,\, A^3_0 \big)
	\end{equation}
	where
	\[
	\begin{split}
	&S^1_0 = \frac{ \Lambda_1}{ L_{2,1}   \varepsilon_S +\mu_1},\\
	&S^2_0 = \frac { L_{2,1} \varepsilon_S (\Lambda_1 + \Lambda_2)
		+\Lambda_2\,\mu_1}{ \varepsilon_S(L_{2,1} L_{3,2}\,\varepsilon_S + L_{2,1}
		\mu_2 + L_{3,2} \mu_1) + \mu_1 \mu_2 },\\
	&S^3_0 = \frac { L_{2,1} L_{3,2} \varepsilon_S^2 (\Lambda_1 +
		\Lambda_2 + \Lambda_3 ) + \varepsilon_S \big[\Lambda_3 L_{2,1} \mu_2 + L_{3,2} \mu_1 (\Lambda_2 + \Lambda_3) \big]
		+\Lambda_3 \mu_1 \mu_2}{\mu_3 \big[ \varepsilon_S (L_{2,1}L_{3,2} \varepsilon_S + L_{2,1} \mu_2
		+ L_{3,2} \mu_1) +\mu_1 \mu_2 \big] },\\
	& I^1_0 = I^2_0 = I^3_0 = C^1_0 = C^2_0 = C^3_0 = A^1_0 = A^2_0 = A^3_0 = 0.
	\end{split}
	\]
	Similarly, the global basic reproduction number of a three-nodes chain
	reveals that the couplings coefficients $L_{2,1}$ and $L_{3,2}$ affect the dynamics of each node
	of the network,
	and are likely to produce undesirable phenomenon.
	But its nebulous expression seems to forbid
	any relevant interpretation.
	However, we are going to see in the next subsection, that the existence of a unique stable disease-free equilibrium
	for the network is guaranteed
	under reasonable assumption.
\end{rema}

\begin{figure}[h!]
	\centering
	\tikzstyle{greencircle} = [circle,
	draw,
	green!40!brown,
	bottom color=green!40!brown,
	top color= white,
	text=black]
	\tikzstyle{redcircle} = [circle,
	draw,
	red!40!brown,
	bottom color=red!40!brown,
	top color= white,	
	text=black]
	\tikzstyle{bluecircle} = [circle,
	draw,
	blue!40!brown,
	bottom color=blue!40!brown,
	top color= white,	
	text=black]
	\begin{tikzpicture}
	\node[greencircle] (x1) at (-2.5,0) {$1$};
	\node[redcircle] (x2) at (2.5,0) {$2$};
	\node[bluecircle] (x3) at (7.5,0) {$3$};
	\draw[>=latex,->] (x1) to  node [pos=0.5,above]{\small $L_{2,1}$}(x2);
	\draw[>=latex,->] (x2) to  node [pos=0.5,above]{\small $L_{3,2}$}(x3);
	\end{tikzpicture}
	\caption{Three-nodes chain built with nonidentical instances of system \eqref{eq:SICA-model}.}
	\label{fig:three-nodes-chain}
\end{figure}
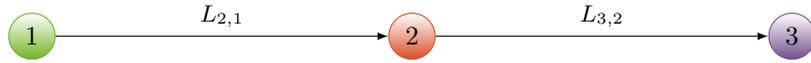

\subsection{Disease-Free Equilibrium of the complex network}

Here, our aim is to overcome the computational difficulties met in the previous subsections.
Thus we establish in the general case that the complex network \eqref{eq:SICA-network} admits
a unique stable disease-free equilibrium under reasonable assumption.

\begin{theo}
	The complex network \eqref{eq:SICA-network} admits a unique disease-free equilibrium $\Sigma_0$,
	which is globally asymptotically stable in the region $\Omega$ defined by \eqref{eq:region-Omega},
	provided 
	\begin{equation}
	\dfrac{\Lambda_0}{\mu_0} \dfrac{\mathcal{N}_i}{\mathcal{D}_i} < 1,
	\label{eq:assumption-unique-DFE-network-problem}
	\end{equation}
	for all $i \in \lbrace 1,\,\dots,\,n \rbrace$, where $\mathcal{N}_i$ and $\mathcal{D}_i$ are defined by
	\[
	\begin{split}
	&\mathcal{N}_i = \beta_i k_{1,i},\\
	&\mathcal{D}_i = \mu_i \big[\xi_{2,i}(\rho_i + \xi_{1,i}) + \phi_i\xi_{1,i} + \rho_id_i \big] + \rho_i \omega_i d_i,
	\end{split}
	\]
	with
	\begin{equation}
	\begin{split}
	&k_{1,i} = \xi_{1,i} \xi_{2,i} + \xi_{1,i}\phi_i\eta_{C,i} + \xi_{2,i} \rho_i \eta_{A,i},\\
	&k_{2,i} = \xi_{1,i} \omega_i + \xi_{1,i}\xi_{3,i} \eta_{C,i} + \rho_i \eta_{A,i}\omega_i - \eta_{C,i}\rho_i\alpha_i,\\
	&k_{3,i} = \alpha_i \xi_{2,i} + \xi_{2,i} \xi_{3,i} \eta_{A,i} + \phi_i\eta_{C,i} \alpha_i - pji_i \eta_{A,i} \omega_i,\\
	&\xi_{1,i} = \alpha_i + \mu_i + d_i,	\\
	&\xi_{2,i} = \omega_i + \mu_i,	\\
	&\xi_{3,i} = \rho_i + \phi_i + \mu_i.
	\end{split}
	\label{eq:coeff-ki-xii}
	\end{equation}
	\label{theo:global-dfe-network}
\end{theo}

\begin{proof}
	The equilibrium points of the network problem \eqref{eq:SICA-network} are the solutions of the system 
	\[
	\dot{S}_j = \dot{I}_j = \dot{C}_j = \dot{A}_j = 0, \quad 1 \leq j \leq n.
	\]
	We determine the disease-free equilibria by assuming that $I_j = C_j = A_j = 0$ for all $j \in \lbrace 1,\, \dots,\, n \rbrace$.
	Then we directly obtain
	\[
	\dot{I}_j = \dot{C}_j = \dot{A}_j = 0, \quad 1 \leq j \leq n,
	\]
	and simultaneously
	\begin{equation}
	\begin{cases}
	\mu_1 S_1 - \varepsilon_S \displaystyle\sum_{k=1}^n L_{1,k}S_k = \Lambda_1,\\
	\mu_2 S_2 - \varepsilon_S \displaystyle\sum_{k=1}^n L_{2,k}S_k = \Lambda_2,\\
	\dots\\
	\mu_n S_n - \varepsilon_S \displaystyle\sum_{k=1}^n L_{n,k}S_k = \Lambda_n.
	\end{cases}
	\label{eq:linear-system-DFE-SICA-network}
	\end{equation}
	The latter system is a linear system which can be written
	\[
	B Y = \Lambda, \quad Y = \big(S_1,\,\dots,\,S_n \big)^T, \quad
	\Lambda = \big( \Lambda_1,\,\dots,\,\Lambda_n \big)^T,
	\]
	with $B = B_1 - \varepsilon_S L$, $L$ being the matrix of connectivity defined as in Section~\ref{section:Construction-of-the-complex-network},
	and $B_1$ is a diagonal matrix storing the mortality rates,
	that is $B_1 = \diag{\mu_1,\, \dots,\,\mu_n}$.
	$L$ being a zero column-sum matrix, it follows that $B$ is a strictly diagonally dominant matrix.
	By virtue of Levy-Desplanques Theorem \cite{horn2012matrix}, $B$ is an invertible matrix.
	Hence, system \eqref{eq:linear-system-DFE-SICA-network} admits a unique solution,
	which corresponds to the unique disease-free equilibrium $\Sigma_0$ of the network problem \eqref{eq:SICA-network}.
	
	Next, we introduce the Lyapunov functional $V$ defined by
	\[
	V = \sum_{i=1}^n V_i,
	\]
	where $V_i$ is the Lyapunov function introduced in \cite{silva2017global} (proof of Theorem 1), given
	by
	\[
	V_i = k_{1,i} I_i + k_{2,i} C_i + k_{3,i} A_i, \quad 1 \leq i \leq n,
	\]
	where the coefficients $k_{1,i}$, $k_{2,i}$ and $k_{3,i}$ are determined by \eqref{eq:coeff-ki-xii}.
	We compute the orbital derivative $\dot{V}$
	of the Lyapunov functional $V$ along a solution $X$ starting in $\Omega$ :
	\[
	\begin{split}
	\dot{V}
	&=	\displaystyle\sum_{i=1}^n \big(k_{1,i} \dot{I}_i + k_{2,i} \dot{C}_i + k_{3,i} \dot{A}_i \big) \\
	&= \sum_{i=1}^n \Big[\big(\mathcal{N}_i I_i S_i - \mathcal{D}_i I_i \big)
	+ \eta_{C,i}\big(\mathcal{N}_i C_i S_i - \mathcal{D}_i C_i \big)
	+ \eta_{A,i}\big(\mathcal{N}_i A_i S_i - \mathcal{D}_i A_i \big)\Big]\\
	&\leq \sum_{i=1}^n \left[\left(\mathcal{N}_i\frac{\Lambda_0}{\mu_0}-\mathcal{D}_i\right)I_i
	+ \eta_{C_i}\left(\mathcal{N}_i\frac{\Lambda_0}{\mu_0}-\mathcal{D}_i\right)C_i
	+ \eta_{A_i}\left(\mathcal{N}_i\frac{\Lambda_0}{\mu_0}-\mathcal{D}_i\right)A_i \right],
	\end{split}
	\]
	which guarantees that $\dot{V} \leq 0$, since we assume that
	$\frac{\Lambda_0}{\mu_0} \frac{\mathcal{N}_i}{\mathcal{D}_i} < 1$ for all $i \in \lbrace 1,\,\dots,\,n \rbrace$.
	
	Finally, it is seen that $\dot{V} = 0$ if and only if $I_i = C_i = A_i = 0$ for all $i \in \lbrace 1,\,\dots,\,n \rbrace$.
	The conclusion follows from LaSalle invariance principle \cite{lasalle1960some}.
\end{proof}

\begin{rema}
	Since we have $\Lambda_i \leq \Lambda_0$ and $\mu_i \geq \mu_0$ for all $i \in \lbrace 1,\,\dots,\,n \rbrace$,
	assumption \eqref{eq:assumption-unique-DFE-network-problem} implies that
	\[
	\dfrac{\Lambda_i}{\mu_i} \dfrac{\mathcal{N}_i}{\mathcal{D}_i} < 1,
	\quad 1 \leq i \leq n.
	\]
	As it is relevant to introduce again $R_{0,i} = \dfrac{\Lambda_i}{\mu_i} \dfrac{\mathcal{N}_i}{\mathcal{D}_i}$
	for each $i \in \lbrace 1,\,\dots,\,n \rbrace$,
	it is seen that assumption \eqref{eq:assumption-unique-DFE-network-problem} is a sufficient condition for the existence
	of a unique stable disease-free equilibrium in the network, which requires that every node in the network has 
	a ``small'' basic reproduction number $R_{0,i}$. If only one node violates this condition, then the network is likely to exhibit
	undesirable equilibrium states. In other words, Theorem~\ref{theo:global-dfe-network} generalizes the pattern
	discovered with a two-nodes network in Proposition~\ref{prop:R0:2nodes}.
\end{rema}

\section{A case study: Cape Verde archipelago}
\label{sec:casestudyCV}

In this section, we study the case of Cape Verde archipelago,
which has been affected by HIV/AIDS epidemics for several decades.
Our aim is to determine a topology which could temper the spreading of the epidemics.

\subsection{Geographical background}

Cape Verde is an archipelago of $10$ volcanic islands, located 
in the Atlantic Ocean, at about 570 kilometers from the Northwest African coast.
Since $1$ of those $10$ islands has no inhabitants,
we propose to model this archipelago with a $9$ nodes network
(see Figure~\ref{fig:cape-verde-archipelago} below).
We assume that the network is divided into $3$ groups of nodes:
group $1$ is composed with nodes $1$, $2$, $3$, $4$, $5$,
group $2$ with nodes $6$, $7$, $8$,
and group $3$ with single node $9$, corresponding to Santiago island which
is the most important island in the archipelago, with the greatest number of HIV infected inhabitants.
The parameters values are given in Table~\ref{tab:parameters-values-cape-verde}.
In absence of coupling, it is relevant to compute the basic reproduction number $R_0$ for each group: 
$R_0 \simeq 0.914$ for group $1$,
$R_0 \simeq 1.371$ for group $2$ and
$R_0 \simeq 7.312$ for group $3$.

\begin{rema}
	The value of the basic reproduction number for group $3$ implies that assumption \eqref{eq:assumption-unique-DFE-network-problem}
	of Theorem~\ref{theo:global-dfe-network}
	may not be fulfilled, which could lead to the emergence of undesirable equilibrium states, with a persistence of the infection
	within the population for instance. Thus it appears crucial to limit the spreading of the infection
	at a reasonable level,
	by finding a suitable topology of the network.
\end{rema}

The coupling strengths are fixed as follows:
\[
\varepsilon_S = 0.02, \quad
\varepsilon_I = 0.01, \quad
\varepsilon_C = 0.01, \quad
\varepsilon_A = 0.01,
\]
in the case of weak coupling, or
\[
\varepsilon_S = 0.2, \quad
\varepsilon_I = 0.3, \quad
\varepsilon_C = 0.1, \quad
\varepsilon_A = 0.3,
\]
in the case of strong coupling.
The initial condition $X_0$ partially corresponds to official data:
approximate values of the total population $N_j(0)$ for each node ($1 \leq j \leq 9$)
in $2015$ can be found in \cite{RapportCapeVerde},
as well as approximate values of infected individuals $I_j(0)$.
The values of $C_j(0)$ and $A_j(0)$ have been assumed, so that
the corresponding subpopulations are in proportionality with the total population.

\begin{figure}[h!]
	\centering
	\includegraphics[scale=1]{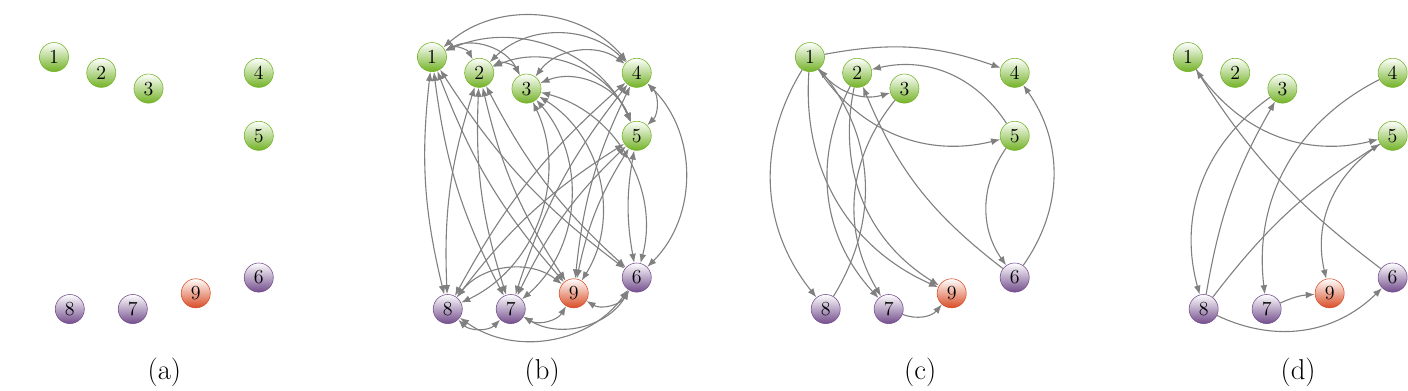}
	\caption{Cape Verde archipelago: four distinct topology sets.
		(a) Empty topology corresponding to a network without couplings.
		(b) Complete graph topology.
		(c) and (d) Weakly dense topologies.}
	\label{fig:cape-verde-archipelago}
\end{figure}

\begin{table}[h!]
	\caption{Cape Verde archipelago modeled by a $9$ nodes complex network.
		Official $2015$ data \cite{RapportCapeVerde} are marked with a star. Other numerical data have been chosen arbitrarily.}
	\medskip
	\centering
	\begin{tabular}{l l l l l l l l}
		\hline\hline
		\textbf{Island}	&	\textbf{Node} 	&	$N(0)$		&	$S(0)$	&	$I(0)$		&	$C(0)$	&	$A(0)$	&	$R_0$	\\
		\hline                                        
		Santo Ant\~ao	&	$1$				&	$40500^*$	&	40388	&	$10^*$		&	93		&	9		&	$0.914$	\\
		S\~ao Vicente	&	$2$				&	$81000^*$	&	80763	&	$32^*$		&	186		&	19		&	$0.914$	\\
		Sau Nicolau		&	$3$				&	$12420^*$	&	12381	&	$7^*$		&	29		&	3		&	$0.914$	\\
		Sal				&	$4$				&	$33750^*$	&	33642	&	$22^*$		&	78		&	8		&	$0.914$	\\
		B\~oa Vista		&	$5$				&	$14450^*$	&	14404	&	$10^*$		&	33		&	3		&	$0.914$	\\
		Maio			&	$6$				&	$6980^*$	&	6957	&	$5^*$		&	16		&	2		&	$1.371$	\\
		Fogo			&	$7$				&	$35840^*$	&	35735	&	$15^*$		&	82		&	8		&	$1.371$	\\
		Brava			&	$8$				&	$5700^*$	&	5681	&	$5^*$		&	13		&	1		&	$1.371$	\\
		Santiago		&	$9$				&	$394130^*$	&	293084	&	$303^*$		&	676		&	67		&	$7.312$ \\
		\hline\hline
	\end{tabular}
	\label{tab:islands-cape-verde}
\end{table}

\begin{table}[h!]
	\caption{Parameters values for the numerical simulations of Cape Verde archipelago.}
	\medskip
	\centering
	\begin{tabular}{c r r r}
		\hline\hline 
		\textbf{Parameter}	&	\textbf{Nodes 1, 2, 3, 4, 5}	&	\textbf{Nodes 6, 7, 8}	&	\textbf{Node 9}		\\
		\hline                                                                                                        
		$\Lambda_i$			&	$2$								&	$2$		  				&	$2$		  	 		\\
		$\beta_i$			&	$0.001$							&	$0.0015$ 		 		&	$0.008$	 		 	\\
		$\eta_{C,i}$		&	$0.04$							&	$0.04$	 		 		&	$0.04$	  	   		\\
		$\eta_{A,i}$		&	$1.3$							&	$1.3$	 				&	$1.3$	  	  		\\
		$\mu_i$				&	$1/70$							&	$1/70$	 		 		&	$1/70$	  	   		\\
		$\rho_i$			&	$0.1$							&	$0.1$	 				&	$0.1$	  	  		\\
		$\phi_i$			&	$1$								&	$1$		 				&	$1$		  	 		\\
		$\omega_i$			&	$0.09$							&	$0.09$	 				&	$0.09$	  	   		\\
		$\alpha_i$			&	$0.33$							&	$0.33$	 				&	$0.33$	  	   		\\
		$d_i$				&	$1$       						&	$1$       				&	$1$ 				\\
		\hline\hline            	       	
	\end{tabular}
	\label{tab:parameters-values-cape-verde}
\end{table}

\subsection{Randomly generated topologies}

The numerical integration on a finite time interval $[0,\,T]$
of the complex network \eqref{eq:SICA-network} modeling Cape Verde archipelago
has been performed using the \texttt{python} language, in a GNU/LINUX environment.
For each set of parameters, let us introduce the final level of infected individuals, given by
\begin{equation}
L_f = \displaystyle\sum_{j=1}^n \big[I_j(T) + C_j(T) + A_j(T) \big].
\label{eq:final-level-infected-individuals}
\end{equation}
In absence of coupling (see Figure~\ref{fig:cape-verde-archipelago}(a)), we obtain
$L_f \simeq 9112.77$ with $T=200$,
whereas the complete graph topology (see Figure~\ref{fig:cape-verde-archipelago}(b)) leads to $L_f \simeq 9161.02$.
Since the couplings are likely to produce emerging equilibria, we propose to explore the possible topologies for the complex
network modeling Cape Verde archipelago.
The set of possible topologies being finite, there obviously exists an optimal topology
minimizing the level of infection $L_f$.
Thus our goal is to determine a near-optimal topology.
However, it is easily seen that a $9$ nodes network can admit at most $72$ edges, assuming
that there are no loops nor parallel edges.
The total number of possible topologies is given by the sum of binomial coefficients
\[
\displaystyle\sum_{k=1}^{72}\binom{72}{k} \simeq 4.72.10^{21},
\]
thus it is not reasonable to explore the total set of topologies.
We propose to investigate a sample of randomly generated topologies,
by choosing a random number of edges $1 \leq \abs{\mathscr{E}} \leq 72$,
and a random subset of $\abs{\mathscr{E}}$ edges.
We have computed the final level $L_f$ of infected individuals for a sample of $1400$ randomly
generated topologies.
The result is depicted in Figure~\ref{fig:capo-verde-random-simulation}, where each
red cross has coordinates $\left(L_f,\,\abs{\mathscr{E}}\right)$. The green dotted vertical line
corresponds to the level of infected individuals for an empty topology.

\begin{figure}
	\centering
	\includegraphics[scale=0.85]{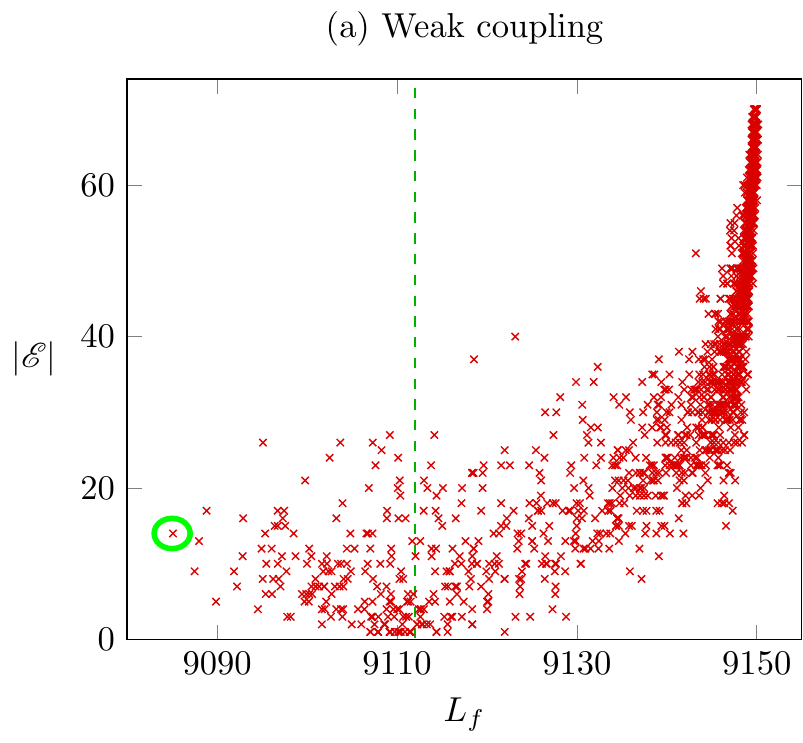}
	\includegraphics[scale=0.85]{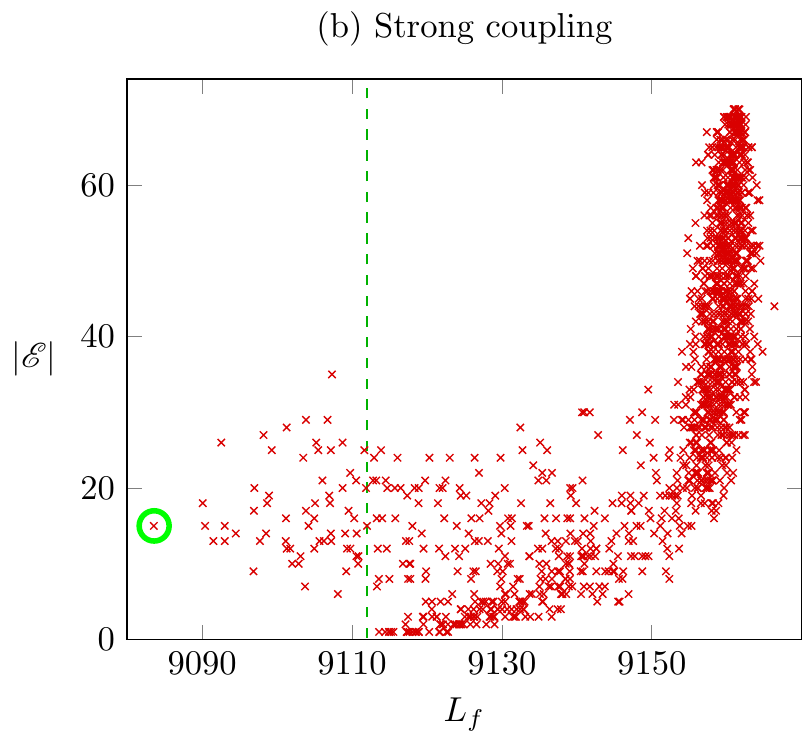}
	\caption{Numerical results for two samples of $1400$ randomly generated topologies modeling Cape Verde (9 islands). 
		The green dotted vertical line of equation $x=9113$ shows the level of infected individuals without coupling.
		The optimal topology is marked with a green circle.
		(a) Weak coupling: $\varepsilon_S = 0.02$,
		$\varepsilon_I = \varepsilon_C = \varepsilon_A = 0.01$.
		(b) Strong coupling: $\varepsilon_S = 0.2$,
		$\varepsilon_I = \varepsilon_A = 0.3$, $\varepsilon_C = 0.1$.}
	\label{fig:capo-verde-random-simulation}
\end{figure}

We observe that the final level of infected individuals $L_f$
varies a lot with respect to the number $\abs{\mathscr{E}}$ of edges.
It seems that a dense topology, with a number of edges neighbor to the maximal number $72$,
corresponding to the complete graph topology, produces a high level of infection.
Meanwhile, a weakly dense topology is not a warranty for a low final level of infection.
However, this random simulation has detected an optimal topology (marked with a green circle in Figure~\ref{fig:capo-verde-random-simulation})
for which the final level of infection is lesser than the benchmark $L_f \simeq 9113$ obtained for an empty topology.
Furthermore, we observe that the two clouds of points obtained for weak or strong coupling roughly admit similar shapes.
In other words, the topology seems to be more important than the coupling strength.

\subsection{Weakly dense topologies}

The random simulation presented in the previous section seems to exclude dense topologies.
The question of how to select a weakly dense topology, in order to temper the final level of infected
individuals $L_f$ remains delicate.
Finally, we present the times series corresponding to two weakly dense topologies.

The first weakly dense topology we aim to analyze is a near-optimal topology
detected by the random simulation (see Figure~\ref{fig:cape-verde-archipelago}(c)); it admits a set of $14$ edges, given by
\[
\mathscr{E} = \Big\lbrace
[1, 3], [2, 9], [5, 6], [3, 7], [7, 9], [2, 7], [1, 9], [6, 2], [6, 4], [2, 8], [5, 2], [8, 1], [1, 5], [1, 4]
\Big\rbrace.
\]
The time series of the corresponding complex network are shown in Figure~\ref{fig:capo-verde-topology-c}.
The final level of infected individuals for that optimal topology is $L_f \simeq 9085.09$.

\begin{figure}[h!]
	\centering
	\includegraphics[scale=0.72]{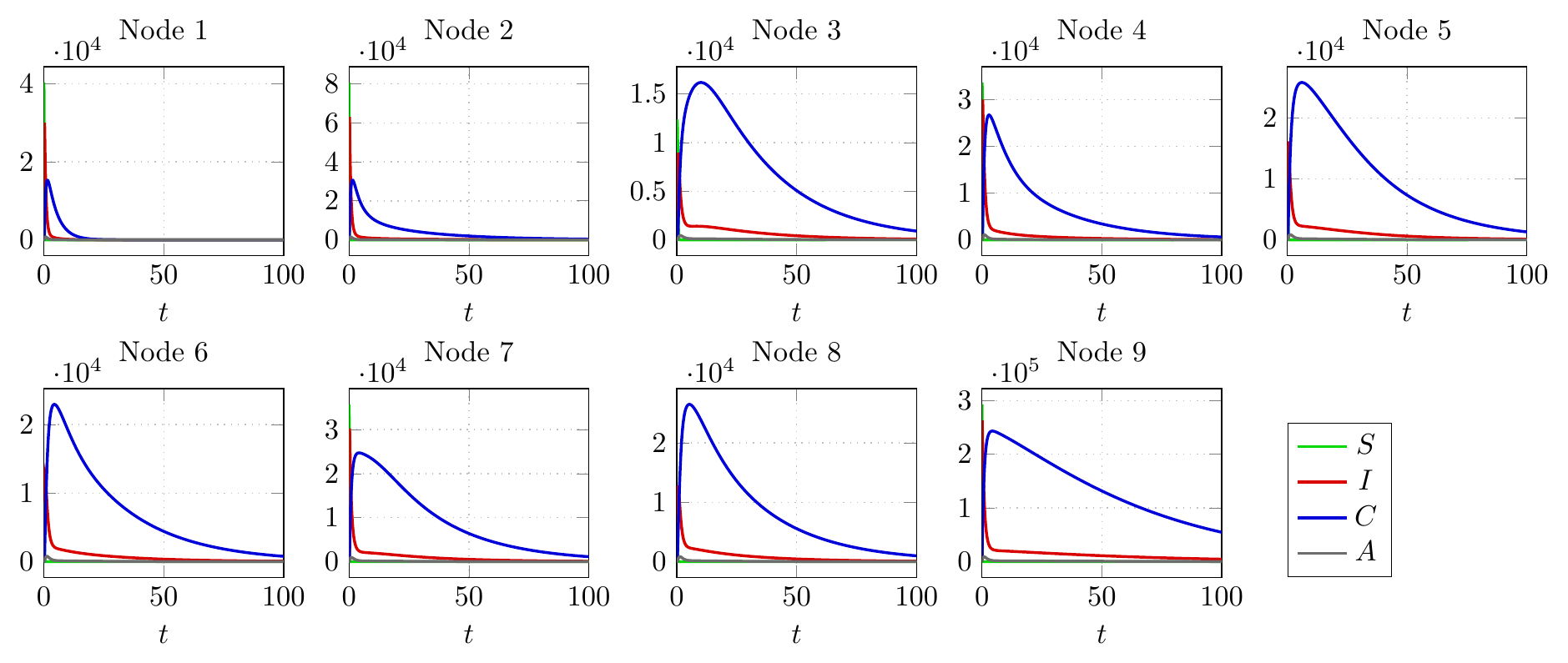}
	\caption{Numerical results for topology (c).}
	\label{fig:capo-verde-topology-c}
\end{figure}

The second weakly dense topology (see Figure~\ref{fig:cape-verde-archipelago}(d)) we focus on
is another near-optimal topology; it admits a set of $9$ edges is given by
\[
\mathscr{E} = \Big\lbrace
[1, 5], [8, 6], [5, 9], [8, 5], [4, 7], [3, 8], [6, 1], [8, 3], [7, 9]
\Big\rbrace.
\]
The time series of the corresponding complex network are shown in Figure~\ref{fig:capo-verde-topology-d}.
The final level of infected individuals for that second weakly dense topology is $L_f \simeq 9087.50$.

\begin{figure}[h!]
	\centering
	\includegraphics[scale=0.72]{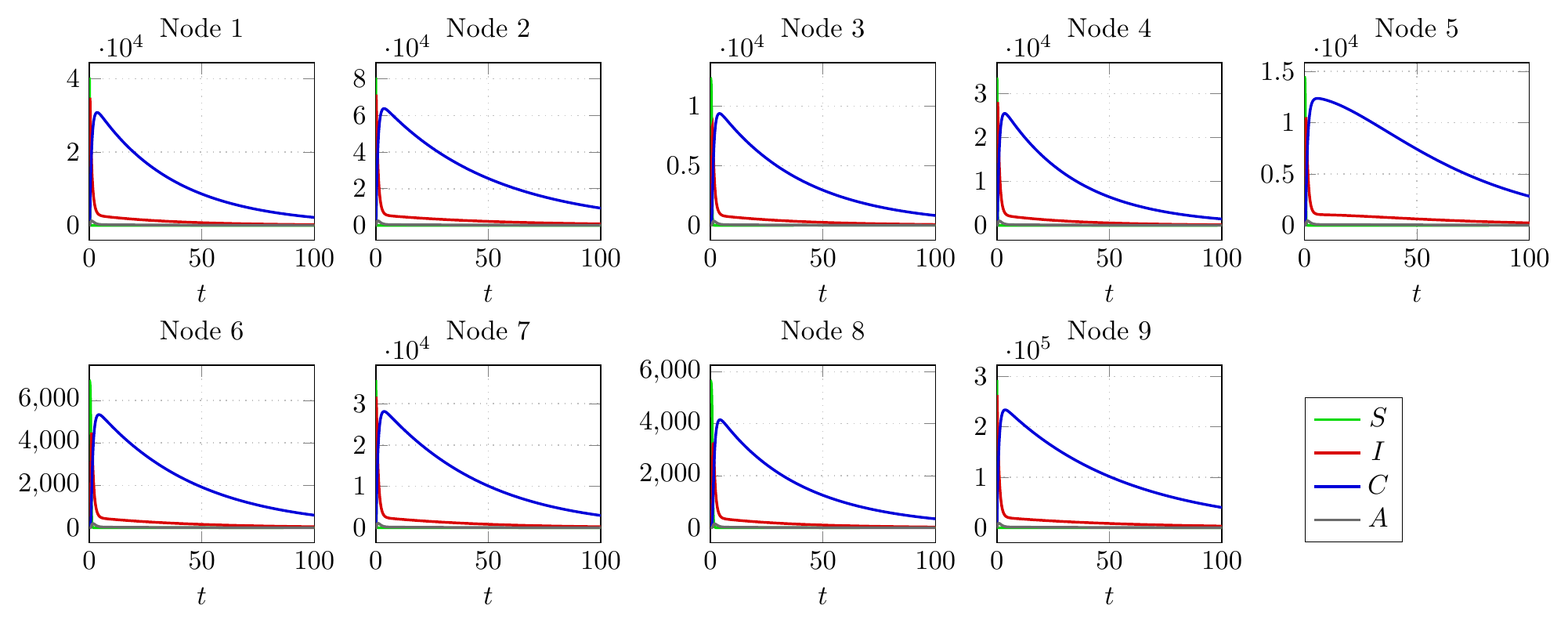}
	\caption{Numerical results for topology (d).}
	\label{fig:capo-verde-topology-d}
\end{figure}

\begin{rema}
	The numerical results presented in this section can help finding a favorable situation for limiting the level of infection
	in Cape Verde archipelago. Indeed, the results of the simulation of randomly generated topologies appear to exclude dense
	topologies, which means that one should avoid important human migrations from one island to another.
	In the mean time, weakly dense topologies (c) and (d) presented in Figure \ref{fig:cape-verde-archipelago}
	seem to favor migrations stemming from islands admitting a small basic reproduction number.
	Nevertheless, those interpretations should be prudently nuanced, since they are the result of a mathematical model
	whose scope is necessarily limited.
\end{rema}

\section{Conclusion}
\label{sec:concl}
In this work, we presented the analysis of a complex network of dynamical systems for the study of the spread of HIV/AIDS epidemics.
Built with nonidentical instances of a compartmental model for which a disease-free equilibrium and an endemic equilibrium
can coexist, this complex network exhibits a positively invariant region and presents a unique disease-free equilibrium
which is globally asymptotically stable,
under the assumption that each node composing the network admits a small basic reproduction number.
However, emerging equilibria are likely to appear if this assumption is not fulfilled,
and we proposed a numerical strategy in order to detect a near-optimal topology for which the level of infection is minimized.
This method has been applied to the case of the Cape Verde archipelago, and we exhibited a near-optimal topology which seems to be robust with respect to a variation of the coupling strength.
However, it seems delicate to identify the characteristic features of such a near-optimal topology, since weakly dense topologies can produce a high level of infection as well as limit the infection at a low level.

In a future work, we aim to deepen this subtle question, which could lead to
establishing a necessary and sufficient condition of synchronization in the network.
In parallel, we propose to improve our model by applying an optimal control process, in order to reach a global disease-free equilibrium,
in spite of the risk that a small group of nodes in the network could admit a high basic reproduction number.
This control process could be introduced at a double scale,
with control actions exerted into the dynamics of each node, and simultaneously control actions exerted along the connections of the network.

As a final perspective, we also intend to study of the effect of introducing delays in the migrations supported by the connections of the network,
since it is likely
to reveal new dynamics which might be hidden at this stage.

\section*{Acknowledgments}

The authors are very grateful to the anonymous reviewers whose comments greatly improved
the presentation of the paper.

This research was partially supported by the Portuguese Foundation for Science and Technology (FCT)
within projects UID/MAT/04106/2019 (CIDMA) and PTDC/EEI-AUT/2933/2014 (TOCCATTA),
co-funded by FEDER funds through COMPETE2020 -- Programa Operacional Competitividade e Internacionaliza\c{c}\~ao (POCI)
and by national funds (FCT). Silva is also supported by national funds (OE), through FCT, I.P.,
in the scope of the framework contract foreseen in the numbers 4, 5 and 6 of the article 23, of the Decree-Law 57/2016,
of August 29, changed by Law 57/2017, of July 19.

\section*{Conflict of Interest}
The authors declare there is no conflict of interest in this paper.

% ---------------------------------------------


\begin{thebibliography}{999}

\bibitem{arenas2008synchronization}
\newblock  A. Arenas, A. D{\'\i}az-Guilera, J. Kurths, Y. Moreno, C. Zhou,
{\em Synchronization in complex networks},
\newblock  Physics Reports,
\newblock  \textbf{469} (2008), 93--153.

\bibitem{aziz2006synchronization}
\newblock  M. A. Aziz-Alaoui, 
{\em Synchronization of Chaos},
\newblock  En\-cyclopedia of Ma\-thema\-tical Phy\-sics,
\newblock  \textbf{5} (2006), 213--226.

\bibitem{belykh2005synchronization}
\newblock  I. Belykh, M. Hasler, M. Lauret, H. Nijmeijer,  
{\em Synchronization and graph topology},
\newblock  International Journal of Bifurcation and Chaos,
\newblock  \textbf{15} (2005), 3423--3433.

\bibitem{cantin2017control}
\newblock  G. Cantin, N. Verdi{\`e}re, V. Lanza, V., et al.
{\em Control of panic behavior in a non identical network coupled with a geographical model},
\newblock In: PhysCon 2017,
\newblock  (2017), 1--6.

\bibitem{cantin2018nonidentical}
\newblock  G. Cantin, 
{\em Non identical coupled networks with a geographical model for human behaviors during catastrophic events},
\newblock International Journal of Bifurcation and Chaos,
\newblock \textbf{27} (2017), 1750213.

\bibitem{cao2006global}
\newblock  J. Cao, P. Li, W. Wang, 
{\em Global synchronization in arrays of delayed neural networks with constant and delayed coupling},
\newblock Physics Letters A,
\newblock \textbf{353} (2006), 318--325.

\bibitem{chen2004global}
\newblock G. Chen, J. Zhou, Z. Liu,  
{\em Global synchronization of coupled delayed neural networks and applications to chaotic CNN models},
\newblock International Journal of Bifurcation and Chaos,
\newblock \textbf{14} (2004), 2229--2240.

\bibitem{driessche2002reproduction}
\newblock P. van den {D}riessche, J. Watmough, 
{\em Reproduction numbers and subthreshold endemic equilibria for compartmental models of disease transmission},
\newblock Math. Biosc.,
\newblock \textbf{180} (2002), 29--40.

\bibitem{epperlein2013phase}
\newblock J. Epperlein, S. Siegmund,  
{\em Phase--locked trajectories for dynamical systems on graphs},
\newblock Discrete \& Continuous Dynamical Systems-B,
\newblock \textbf{18} (2013), 1827--1844.

\bibitem{golubitsky2006nonlinear}
\newblock M. Golubitsky, I. Stewart,  
{\em Nonlinear dynamics of networks: the groupoid formalism},
\newblock Bulletin of the American Mathematical Society,
\newblock \textbf{43} (2006), 305--364.

\bibitem{hethcote2000mathematics}
\newblock H. W. Hethcote,
{\em The mathematics of infectious diseases},
\newblock SIAM review,
\newblock \textbf{42} (2000), 599--653.

\bibitem{horn2012matrix}
\newblock R. A. Horn, C. R. Johnson, {\em Matrix analysis},
\newblock Cambridge {U}niversity {P}ress, 2012. 

\bibitem{keeling2005networks}
\newblock M. J. Keeling, K.TD Eames, 
{\em Networks and epidemic models},
\newblock Journal of the Royal Society Interface,
\newblock \textbf{2} (2005), 295--307.

\bibitem{khanna2014what}
\newblock A.S. Khanna, D.T. Dimitrov, S.M. Goodreau, 
{\em What can mathematical models tell us about the relationship between circular migrations and {HIV} transmission dynamics?},
\newblock Math. Biosci. Eng.,
\newblock \textbf{11} (2014), 1065--1090.

\bibitem{kimbir2012amathematical}
\newblock R.A. Kimbir,  MJ I. Udoo, T. Aboiyar,  
{\em A mathematical model for the transmission dynamics of {HIV}/{AIDS} in a two-sex population considering counseling and antiretroviral therapy ({ART})},
\newblock J. Math. Comput. Sci.,
\newblock \textbf{2} (2012), 1671--1684.

\bibitem{lasalle1960some}
\newblock J. LaSalle, 
{\em Some extensions of {L}iapunov's second method},
\newblock IRE Transactions on circuit theory,
\newblock \textbf{7} (1960), 520--527.

\bibitem{lakshmikantham1989stability}
\newblock V. Lakshmikantham, S. Leela, A. Martynyuk, {\em Stability analysis of nonlinear systems},
\newblock Springer, 1989. 

\bibitem{lassig2001shape}
\newblock M. L{\"a}ssig, U. Bastolla, S.C. Manrubia, A. Valleriani,  
{\em Shape of ecological networks},
\newblock Physical Review Letters,
\newblock \textbf{86} (2001), 4418.

\bibitem{li2010global}
\newblock M.Y. Li, Z. Shuai, 
{\em Global--stability problem for coupled systems of differential equations on networks},
\newblock Journal of Differential Equations,
\newblock \textbf{248} (2010), 1--20.

\bibitem{may2001infection}
\newblock R.M. May, A.L. Lloyd,  
{\em Infection dynamics on scale--free networks},
\newblock Physical Review E,
\newblock \textbf{64} (2001), 066112.

\bibitem{moreno2002epidemic}
\newblock Y. Moreno, R. Pastor--Satorras, A. Vespignani,  
{\em Epidemic outbreaks in complex heterogeneous networks},
\newblock The European Physical Journal B-Condensed Matter and Complex Systems,
\newblock \textbf{26} (2002), 521--529.

\bibitem{podder2011mathematical}
\newblock C.N. Podder, O. Sharomi, A.B. Gumel, E. Strawbridge, 
{\em Mathematical analysis of a model for assessing the impact of antiretroviral therapy, voluntary testing and condom use in curtailing the spread of {HIV}},
\newblock Differ. Equ. Dyn. Syst.,
\newblock \textbf{19} (2011), 283--302.

\bibitem{rahman2016impact}
\newblock S.M.A. Rahman, N.K. Vaidya, X. Zou,  
{\em Impact of early treatment programs on {HIV} epidemics: an immunity-based mathematical model},
\newblock Math. Biosci.,
\newblock \textbf{280} (2016), 38--49.

\bibitem{RapportCapeVerde}
\newblock Comit\'e de Coordena\c{c}\~ao do Combate a Sida, {\em Rapport de Progr\`es de la riposte {VIH/SIDA} {C}abo {V}erde},
\newblock 2015. 

\bibitem{rocha2018effect}
\newblock E.M. Rocha, C.J. Silva, D.F.M. Torres,  
{\em The effect of immigrant communities coming from higher	incidence tuberculosis regions to a host country},
\newblock Ric. Mat.,
\newblock \textbf{67} (2018), 89--112.

\bibitem{silva2017EcoComplexity}
\newblock 
{\em A {SICA} compartmental model in epidemiology with application to {HIV}/{AIDS} in {C}ape {V}erde},
\newblock Ecological Complexity,
\newblock \textbf{30} (2017), 70--75.

\bibitem{silva2017global}
\newblock C.J. Silva,  D.F.M. Torres,  
{\em Global stability for a {HIV}/{AIDS} model},
\newblock Commun. Fac. Sci. Univ. Ank.,
\newblock \textbf{1} (2018), 93--101.

\bibitem{zhang2009multiple}
\newblock C. Zhang, B. Zheng, L. Wang,  
{\em Multiple {H}opf bifurcations of symmetric {BAM} neural network model with delay},
\newblock Applied Mathematics Letters,
\newblock \textbf{22} (2009), 616--622.

\bibitem{WHO}
\newblock WHO, {\em {HIV/AIDS}, Fact sheets},
\newblock 19 July 2018. 

\end{thebibliography}
\end{document}